\documentclass[journal]{IEEEtran}
\usepackage[absolute, showboxes]{textpos}
\TPMargin{5pt}
\newcommand{\copyrightstatement}{
\begin{textblock}{0.80}(0.10, 0.01) 
\noindent
\footnotesize
\copyright 2021 IEEE.
Personal use of this material is permitted.
Permission from IEEE must be obtained for all other uses, in any current or future media, including reprinting/republishing this material for advertising or promotional purposes, creating new collective works, for resale or redistribution to servers or lists, or reuse of any copyrighted component of this work in other works.
\end{textblock}
}
\usepackage{amsmath, amssymb, amsfonts, amsthm, mathtools, bm}
\allowdisplaybreaks
\usepackage{caption, subcaption}
\usepackage{color, cite}
\usepackage{multirow, multicol}
\usepackage{comment}
\usepackage[ruled,vlined]{algorithm2e}
\usepackage{setspace}
\newlength\mylen
\newcommand\myinput[1]{ \settowidth\mylen{\KwIn{}} \setlength\hangindent{\mylen} \hspace*{\mylen}#1\\}
\theoremstyle{definition}
\newtheorem{theorem}{Theorem}

\newtheorem{proposition}[theorem]{Proposition}
\newtheorem{lemma}[theorem]{Lemma}

\newtheorem{remark}{Remark}
\newtheorem{assumption}{Assumption}
\DeclareMathOperator*{\GL}{GL}
\DeclareMathOperator*{\image}{Im}
\DeclareMathOperator*{\trace}{tr}

\DeclareMathOperator*{\diag}{diag}

\DeclareMathOperator*{\const}{const.}
\DeclareMathOperator*{\argmin}{argmin}

\DeclareMathOperator*{\minimize}{minimize}

\DeclareMathOperator*{\subject-to}{subject~to}
\DeclareMathOperator*{\s}{s}
\DeclareMathOperator*{\z}{z}

\begin{document}
\copyrightstatement
\title{Block Coordinate Descent Algorithms for Auxiliary-Function-Based Independent Vector Extraction}
\author{
	Rintaro~Ikeshita, \IEEEmembership{Member, IEEE},
	Tomohiro~Nakatani, \IEEEmembership{Fellow, IEEE},
	and~Shoko~Araki, \IEEEmembership{Senior Member, IEEE}%
	\thanks{
		Rintaro Ikeshita, Tomohiro Nakatani, and Shoko Araki
		are with NTT Corporation, Kyoto, 619-0237, Japan (e-mail: ikeshita@ieee.org).
	}
	\thanks{
	}
}


\maketitle
\begin{abstract}
In this paper, we address the problem of extracting all super-Gaussian source signals from a linear mixture in which
(i) the number of super-Gaussian sources $K$ is less than that of sensors $M$, and
(ii) there are up to $M - K$ stationary Gaussian noises that do not need to be extracted.
To solve this problem,
independent vector extraction (IVE) using a majorization minimization and block coordinate descent (BCD) algorithms has been developed, attaining robust source extraction and low computational cost.
We here improve the conventional BCDs for IVE by carefully exploiting the stationarity of the Gaussian noise components.
We also newly develop a BCD for a semiblind IVE in which the transfer functions for several super-Gaussian sources are given a priori.
Both algorithms consist of a closed-form formula and a generalized eigenvalue decomposition.
In a numerical experiment of extracting speech signals from noisy mixtures,
we show that when $K = 1$ in a blind case or at least $K - 1$ transfer functions are given in a semiblind case,
the convergence of our proposed BCDs is significantly faster
than those of the conventional ones.
\end{abstract}

\IEEEpeerreviewmaketitle

\section{Introduction}
\label{sec:intro}


\IEEEPARstart{B}{lind} source separation (BSS) is a problem concerned with estimating the original source signals from a mixture signal captured by multiple sensors.
When the number of sources is no greater than that of sensors, i.e., in the (over-)determined case,
independent component analysis (ICA~\cite{comon2010handbook,cichocki2002adaptive,lee1998independent,stone2004independent})
has been a popular approach to BSS because it is simple and effective.
When each original source is a vector and a multivariate random variable,
independent vector analysis (IVA~\cite{kim2007,hiroe2006}), also termed joint BSS~\cite{li2009joint,li2011joint,anderson2011joint}, has been widely studied as an extension of ICA.

In this paper, we will only focus on the problem of extracting all the super-Gaussian sources from a linear mixture signal under the following assumptions and improve the computational efficacy of IVA (ICA can also be improved in the same way as IVA, and so we only deal with IVA):
\begin{enumerate}
\item The number of super-Gaussian sources $K$ is known and fewer than that of sensors $M$, i.e., $K < M$.
\item There can be up to $M - K$ stationary Gaussian noises, and thus the problem remains (over-)determined.
\end{enumerate}
The second assumption, which concerns the rigorous development of efficient algorithms, can be violated to some extent when applied in practice (see numerical experiments in Section~\ref{sec:exp}).
To distinguish it from a general BSS, this problem is called blind source extraction (BSE~\cite{cichocki2002adaptive}).
The BSE problem is often encountered in such applications as speaker source enhancement, biomedical signal processing, or passive radar/sonar.
In speaker source enhancement, for instance, 
the maximum number of speakers (super-Gaussian sources) can often be predetermined as a certain number (e.g., two or three) while an audio device is equipped with more microphones.
In real-word applications, the observed signal is contaminated with background noises, and most noise signals are more stationary than speaker signals.

BSE can be solved by simply applying IVA as if there were $M$ super-Gaussian sources and selecting the top $K$ $(< M)$ super-Gaussian signals from the $M$ separated signals in some way.
However, this approach is computationally intensive if $M$ gets too large.
To reduce the computing cost, for preprocessing of IVA, the number of channels (or sensors) can be reduced to $K$ by using principle component analysis or by selecting $K$ sensors with high SNRs.
These channel reductions, however, often degrade the separation performance due to the presence of the background noises.

BSE methods for efficiently extracting just one or several non-Gaussian signals (not restricted to super-Gaussian signals) have already been studied mainly in the context of ICA%
~\cite{friedman1974,huber1985,friedman1987exploratory,cardoso1993,delfosse1995adaptive,cruces2004,amari1998adaptive,hyvarinen1997fastica,hyvarinen1999fastica,oja2006fastica,wei2015}.
The natural gradient algorithm and the FastICA method with deflation techniques~\cite{amari1998adaptive,hyvarinen1997fastica,hyvarinen1999fastica} can sequentially extract non-Gaussian sources one by one.
In FastICA, the deflation is based on the orthogonal constraint where the sample correlation between every pair of separated signals equals zero.
This constraint was also used to develop FastICA with symmetric orthonormalization~\cite{hyvarinen1997fastica,hyvarinen1999fastica,oja2006fastica,wei2015} that can simultaneously extract $K$ non-Gaussian signals.

\subsection{Independent vector extraction (IVE)}

Recently, maximum likelihood approaches have been proposed for BSE in which the background noise components are modeled as stationary Gaussians.
These methods include
independent vector extraction (IVE~\cite{koldovsky2018ive,koldovsky2017ive,jansky2020adaptive})
and overdetermined IVA (OverIVA~\cite{scheibler2019overiva,scheibler2020fast,scheibler2020ive,ike2020overiva}),
which will be collectively referred to as IVE in this paper.
When $K$ non-Gaussian sources are all super-Gaussian
(as defined in Assumption~\ref{assumption:superGaussian}), IVE can use a majorization-minimization (MM~\cite{lange2016mm}) optimization algorithm developed for
an auxiliary-function-based ICA (AuxICA~\cite{ono2010auxica})
and
an auxiliary-function-based IVA (AuxIVA~\cite{ono2011auxiva,ono2012auxiva-stereo,ono2018asj}).
In this paper, we only deal with an IVE based on the MM algorithm~\cite{scheibler2019overiva,scheibler2020fast,scheibler2020ive,ike2020overiva,jansky2020adaptive}
and always assume that the super-Gaussian distributions are defined by Assumption~\ref{assumption:superGaussian} in Section~\ref{sec:model:IVE}.

In the MM-based IVE, 
each separation matrix $W \in \mathbb{C}^{M \times M}$ is optimized by iteratively minimizing a surrogate function of the following form:
\begin{align*}
g(W) &=
 \sum_{i = 1}^K \bm{w}_i^h V_i \bm{w}_i + \trace\left( W_{\z}^h V_{\z} W_{\z} \right) - \log | \det W |^2,
\\
W &= [\bm{w}_1,\ldots,\bm{w}_K,W_{\z}] \in \mathbb{C}^{M \times M},
\\
W_{\z} &= [\bm{w}_{K+1},\ldots,\bm{w}_M] \in \mathbb{C}^{M \times (M - K)}.
\end{align*}
Here, $\bm{w}_1,\ldots,\bm{w}_K \in \mathbb{C}^M$ are filters that extract the target-source signals and $W_{\z}$ is a filter that extracts the $M - K$ noise components.
$V_1,\ldots,V_K,V_{\z} \in \mathbb{C}^{M \times M}$ are the Hermitian positive definite matrices that are updated in each iteration of the MM algorithm.
Interestingly, when $K = M$ or when viewing the second term of $g(W)$ as $\trace\left( W_{\z}^h V_{\z} W_{\z} \right) = \sum_{i = K + 1}^M \bm{w}_i^h V_{\z} \bm{w}_i$,
the problem of minimizing $g(W)$ has been discussed in ICA/IVA literature~\cite{pham2001,degerine2004,degerine2006maxdet,yeredor2009HEAD,yeredor2012SeDJoCo}.
Among the algorithms developed so far, block coordinate descent (BCD~\cite{tseng2001convergence}) methods with simple analytical solutions have attracted much attention in the field of audio source separation because they have been experimentally shown to give stable and fast convergence behaviors.
A family of these BCD algorithms~\cite{ono2010auxica,ono2011auxiva,ono2012auxiva-stereo,ono2018asj}, summarized in Table~\ref{table:alg}, is currently called an iterative projection (IP) method.
The IP method has been widely applied not only to ICA and IVA but also to audio source separation methods using more advanced source models (see~\cite{kitamura2016ilrma,kameoka2019MVAE,makishima2019independent,sekiguchi2019fast,ito2019fastmnmf} for the details of such methods).

\subsection{Conventional BCD (or IP) algorithms for IVE}

When we consider directly applying the IP methods developed for AuxIVA to the BSE problem of minimizing $g(W)$ with respect to $W$,
AuxIVA-IP1~\cite{ono2011auxiva} in Table~\ref{table:alg}, for instance, will cyclically update
$\bm{w}_1 \rightarrow \cdots \rightarrow \bm{w}_K \rightarrow \cdots \rightarrow \bm{w}_M$ one by one.
However, in the BSE problem, the signals of interests are $K$ non-Gaussian sources, and most of the optimization of $W_{\z}$ should be skipped.

Therefore, in a previous work \cite{scheibler2019overiva},
an algorithm (IVE-OC in Table~\ref{table:alg}) was proposed that cyclically updates $\bm{w}_1 \rightarrow W_{\z} \rightarrow \cdots \rightarrow \bm{w}_K \rightarrow W_{\z}$ one by one with a computationally cheap formula for $W_{\z}$.
In this work \cite{scheibler2019overiva}, the updating equation for $W_{\z}$ was derived solely from the (weak) orthogonal constraint (OC~\cite{koldovsky2017ive,koldovsky2018ive,jansky2020adaptive}) where the sample correlation between the separated $K$ non-Gaussian sources and the $M - K$ background noises equals zero.
(Note that the $K$ non-Gaussian sources are not constrained to be orthogonal in the OC.)
Although the algorithm has successfully reduced the computational cost of IVA by nearly a factor of $K / M$,
the validity of imposing the heuristic OC remains unclear.

After that, IVE-OC has been extended by removing the OC from the model.
One such extension is a direct method that can obtain a global minimizer of $g(W)$ when $K = 1$~\cite{scheibler2020fast,ike2020overiva}.
The other extension is a BCD algorithm for $K \geq 2$ that cyclically updates the pairs
$(\bm{w}_1,W_{\z}) \rightarrow \cdots \rightarrow (\bm{w}_K,W_{\z})$ one by one~\cite{scheibler2020ive},
but the computational cost is not so cheap due to the full update of $W_{\z}$ in each iteration.
These algorithms are called IVE-IP2 in this paper (see Table~\ref{table:alg}).

\subsection{Contributions}

In this work, we propose BCD algorithms for IVE, which are summarized in Table~\ref{table:alg} with comparisons to the previous BCDs (IPs).
The followings are the contributions of this paper.

(i)
We speed up the previous IVE-IP2 for $K \geq 2$ by showing that 
$W_{\z}$'s update can be omitted in each iteration of BCD without changing the behaviors of the algorithm (Section~\ref{sec:IVE-IP2-new:K>1}).
This is attained by carefully exploiting the stationary Gaussian assumption for the noise components.
In an experiment of speaker source enhancement, we confirmed that the computational cost of the conventional IVE-IP2 is consistently reduced.

(ii)
We provide a comprehensive explanation of IVE-IP2 for $K = 1$ (Sections~\ref{sec:IVE-IP2:K=1} and \ref{sec:source-extraction}).
Interestingly, it turns out to be a method that iteratively updates the maximum signal-to-noise-ratio (MaxSNR~\cite{vantrees2004,warsitz2007maxsnr}) beamformer
and the power spectrum for the unique ($K = 1$) target-source signal.
The experimental result shows that this algorithm has much faster convergence than conventional algorithms.
Note that IVE-IP2 for $K = 1$ was developed independently and simultaneously in our conference paper~\cite{ike2020overiva} and by Scheibler--Ono~\cite{scheibler2020fast}.

(iii)
We reveal that IVE-OC~\cite{scheibler2019overiva}, which was
developed with the help of the heuristic orthogonal constraint (OC),
can also be obtained as a BCD algorithm for our proposed IVE without the OC (Section~\ref{sec:IVE-IP1}).
(This interpretation of IVE-OC as a BCD algorithm was described in our conference paper~\cite{ike2020overiva}, but it is not mathematically rigorous.
We here provide a rigorous proof of this interpretation.)

(iv)
We further extend the proposed IVE-IP2 for the \textit{semiblind} case where the transfer functions for $L$ ($1 \leq L \leq K$) sources of interests can be given a priori (Section~\ref{sec:semi-BSE}).
We call the proposed semiblind method Semi-IVE.
In Semi-IVE, $L$ separation filters, e.g., $\bm{w}_1,\ldots,\bm{w}_L$, which correspond to the known transfer functions are optimized by iteratively solving the linear constrained minimum variance (LCMV~\cite{vantrees2004}) beamforming algorithm,
and the remaining $\bm{w}_{L + 1},\ldots,\bm{w}_K$ (and $W_{\z}$) are optimized by the full-blind IVE-IP2 algorithm.
We experimentally show that when $L \geq K - 1$ transfer functions are given Semi-IVE yields surprisingly fast convergence.

\textit{Organization}: The rest of this paper is organized as follows.
The BSE and semi-BSE problems are defined in Section~\ref{sec:problem}.
The probabilistic model for the proposed IVE is compared with related methods in Section~\ref{sec:model}.
The optimization algorithms are developed separately for the BSE and semi-BSE problems in Sections~\ref{sec:BSE} and \ref{sec:semi-BSE}.
The computational time complexity for these methods is discussed in Section~\ref{sec:computational-complexity}.
In Sections~\ref{sec:exp} and \ref{sec:conclusion}, experimental results and concluding remarks are described.

\subsection{Notations}
Let $\mathcal{S}_+^d$ or $\mathcal{S}_{++}^d$ denote the set of all Hermitian positive semidefinite (PSD) or positive definite (PD) matrices of size $d \times d$.
Let $\GL(d)$ denote the set of all (complex) nonsingular matrices of size $d \times d$.
Also, let $\bm{0}_d \in \mathbb{C}^d$ be the zero vector,
let $O_{i,j} \in \mathbb{C}^{i \times j}$ be the zero matrix,
let $I_d \in \mathbb{C}^{d \times d}$ be the identity matrix,
and let $\bm{e}_k$ be a vector whose $k$-th element equals one and the others are zero. 
For a vector or matrix $A$, $A^\top$ and $A^h$ represent
the transpose and the conjugate transpose of $A$.
The element-wise conjugate is denoted as $A^\ast = (A^\top)^h$.
For a matrix $A \in \mathbb{C}^{i \times j}$, $\image A$ is defined as the subspace $\{ A \bm{x} \mid \bm{x} \in \mathbb{C}^{j} \} \subset \mathbb{C}^i$.
For a vector $\bm{v}$, $\| \bm{v} \| = \sqrt{\bm{v}^h \bm{v}}$ denotes the Euclidean norm.

\begin{table*}[t]
\begin{center}
{\small
\caption{Optimization process of BCD algorithms for problem~\eqref{problem:maxdet}}
\label{table:alg}
\begin{tabular}{cccccl} \hline
& Method & \multicolumn{1}{c}{Algorithm} & Reference & Assumption & \multicolumn{1}{c}{Optimization process of BCD}
\\ \hline
\multirow{3}{*}{Conventional} &
\multirow{3}{*}{(Aux)IVA} & IP1 & \cite{ono2011auxiva} & - &
$\bm{w}_1 \rightarrow \cdots \rightarrow \bm{w}_K \rightarrow \cdots \rightarrow \bm{w}_M$
\\
&& IP2 & \cite{ono2012auxiva-stereo} & $K = M = 2$ &
$W = [ \bm{w}_1,\bm{w}_2 ]$ (direct optimization)
\\
&& IP2 & \cite{ono2018asj} & (if $M$ is even) &
$(\bm{w}_1, \bm{w}_2) \rightarrow \cdots \rightarrow (\bm{w}_{M - 1}, \bm{w}_M)$
\\
\hline
\multirow{2}{*}{Conventional} &
\multirow{2}{*}{IVE} & IVE-OC$\empty^{1)}$ & \cite{scheibler2019overiva}, \S\ref{sec:model:IVE-OC} & Orthogonal constraint (OC) &
$\bm{w}_1 \rightarrow W_{\z}^1 \rightarrow \cdots \rightarrow \bm{w}_K \rightarrow W_{\z}^1$
\\
&& IP2$\empty^{3)}$ & \cite{scheibler2020ive}, \S\ref{sec:IVE-IP2-old:K>1} & $K \geq 2$ & $(\bm{w}_1, W_{\z}) \rightarrow \cdots \rightarrow (\bm{w}_K, W_{\z})$
\\
\hline
\multirow{3}{*}{Proposed} &
\multirow{3}{*}{IVE} & IP1$\empty^{1)}$ & \S\ref{sec:IVE-IP1} & - &
$\bm{w}_1 \rightarrow (W_{\z},\Omega_{\z}) \rightarrow \cdots \rightarrow \bm{w}_K \rightarrow (W_{\z}, \Omega_{\z})$
\\
&& IP2$\empty^{2)}$ & \S\ref{sec:IVE-IP2:K=1} & $K = 1$ &
$W = [\bm{w}_1, (W_{\z})]$ (direct optimization)
\\
&& IP2$\empty^{3)}$ & \S\ref{sec:IVE-IP2-new:K>1} & $K \geq 2$ &
$(\bm{w}_1, (W_{\z})) \rightarrow \cdots \rightarrow (\bm{w}_K, (W_{\z}))$
\\
\hline
\multirow{2}{*}{Proposed} &
\multirow{2}{*}{Semi-IVE} & \multirow{2}{*}{IP2$\empty^{4)}$} & \S\ref{sec:semi-BSE} & Given $\bm{a}_1,\cdots,\bm{a}_L$ $(L \geq K - 1)$ &
$W = [\bm{w}_1, \ldots, \bm{w}_K, (W_{\z})]$ (direct optimization)
\\
&&& \S\ref{sec:semi-BSE} & Given $\bm{a}_1,\cdots,\bm{a}_{L}$ $(L \leq K - 2)$ &
$(\bm{w}_{L + 1}, (W_{\z})) \rightarrow \cdots \rightarrow (\bm{w}_K, (W_{\z}))$
\\
\hline
\multicolumn{6}{l}{$\empty^{1)}$ The two algorithms are identical as shown in Section~\ref{sec:IVE-IP1}.}
\\
\multicolumn{6}{l}{$\empty^{2)}$ It was developed independently and simultaneously by Scheibler--Ono~\cite{scheibler2020fast} and the authors~\cite{ike2020overiva} in Proceedings of ICASSP2020.}
\\
\multicolumn{6}{l}{$\empty^{3)}$ The proposed IVE-IP2 for $K \geq 2$ is an acceleration of the conventional IVE-IP2 developed in~\cite{scheibler2020ive}.}
\\
\multicolumn{6}{l}{$\empty^{4)}$ $\bm{w}_1,\cdots,\bm{w}_L$, which correspond to $\bm{a}_1,\ldots,\bm{a}_L$, are directly globally optimized as the LCMV beamformers (see Section~\ref{sec:semi-BSE:LCMV}).}
\\
\multicolumn{6}{l}{$\empty^{2, 3, 4)}$ In the proposed IVE-IP2 and Semi-IVE, the optimizations for $W_{\z}$ are skipped (see Sections~\ref{sec:IVE-IP2:K=1}, \ref{sec:IVE-IP2-new:K>1}, and \ref{sec:Semi-IVE}).}
\end{tabular}
}
\end{center}
\end{table*}

\section{Blind and semiblind source extraction}
\label{sec:problem}

\subsection{Blind source extraction (BSE) problem}

Throughout this paper, we discuss IVA and IVE using the terminology from audio source separation in the short-term Fourier transform (STFT) domain.

Suppose that $K$ super-Gaussian target-source signals and a stationary Gaussian noise signal of dimension $M - K$ are transmitted and observed by $M$ sensors.
In this paper, we only consider the case where $1 \leq K < M$.
The observed signal in the STFT domain is modeled at each frequency bin $f = 1,\ldots, F$ and time-frame $t = 1,\ldots,T$ as
\begin{align}
\label{eq:mixture}
\bm{x}(f,t) &= A_{\s}(f) \bm{s}(f,t) + A_{\z}(f) \bm{z}(f,t) \in \mathbb{C}^{M},
\\
A_{\s}(f) &= [\, \bm{a}_{1}(f), \ldots, \bm{a}_{K}(f) \,] \in \mathbb{C}^{M \times K},
\\
\bm{s}(f,t) &= [\, s_{1}(f,t), \ldots, s_{K}(f,t) \,]^\top \in \mathbb{C}^K,
\\
A_{\z}(f) &\in \mathbb{C}^{M \times (M - K)}, \quad \bm{z}(f,t) \in \mathbb{C}^{M - K},
\end{align}
where $s_i(f,t) \in \mathbb{C}$ and $\bm{z}(f,t) \in \mathbb{C}^{M - K}$ are the STFT coefficients of target source $i \in \{1,\ldots,K\}$ and the noise signal, respectively.
$\bm{a}_i(f) \in \mathbb{C}^M$ is the (time-independent) acoustic transfer function (or steering vector) of source $i$ to the sensors,
and $A_{\z}(f) \in \mathbb{C}^{M \times (M - K)}$ is that of the noise.
It is assumed that the source signals are statistically mutually independent.

In the \textbf{blind source extraction (BSE)} problem, we are given an observed mixture $\bm{x} = \{ \bm{x}(f,t) \}_{f,t}$ and the number of target sources $K$.
From these inputs, we seek to estimate the spatial images $\bm{x}_1,\ldots,\bm{x}_K$ for the target sources, which are defined as
$\bm{x}_i(f,t) \coloneqq \bm{a}_i(f) s_i(f,t) \in \mathbb{C}^M$, $i = 1,\ldots,K$.

\subsection{Semiblind source extraction (Semi-BSE) problem}

In the \textbf{semiblind source extraction (Semi-BSE)} problem, in addition to the BSE inputs, we are given $L$ transfer functions $\bm{a}_1, \ldots, \bm{a}_L$ for $L$ super-Gaussian sources, where $1 \leq L \leq K$.
From these inputs, we estimate all the target-source spatial images $\bm{x}_1,\ldots,\bm{x}_K$.
If $L = K$, then Semi-BSE is known as a beamforming problem.

\textbf{Motivation to address Semi-BSE:}
In some applications of audio source extraction, such as meeting diarization~\cite{ito2017data}, the locations of some (but not necessarily all) point sources are available or can be estimated accurately,
and their acoustic transfer functions can be obtained from, e.g., the sound propagation model~\cite{johnson1992array}.
For instance, in a conference situation, the attendees may be sitting in chairs with fixed, known locations.
On the other hand, the locations of moderators, panel speakers, or audience may change from time to time and cannot be determined in advance.
In such a case, by using these partial prior knowledge of transfer functions, Semi-BSE methods can improve the computational efficiency and separation performance of BSE methods.
In addition, since there are many effective methods for estimating the transfer function of (at least) a dominant source~\cite{warsitz2007blind,khabbazibasmenj2012robust,vorobyov2013principles},
there is a wide range of applications where Semi-BSE can be used to improve the performance of BSE.

\section{Probabilistic models}
\label{sec:model}

We start by presenting the probabilistic models for the proposed auxiliary-function-based independent vector extraction (IVE),
which is almost the same as those of such related methods as
IVA~\cite{kim2007}, AuxIVA~\cite{ono2011auxiva,ono2012auxiva-stereo,ono2018asj}, and the conventional IVE~\cite{koldovsky2017ive,koldovsky2018ive,jansky2020adaptive,scheibler2019overiva,scheibler2020fast,scheibler2020ive}.

\subsection{Probabilistic model of proposed IVE}
\label{sec:model:IVE}

In the mixing model~\eqref{eq:mixture}, the mixing matrix 
\begin{align}
\label{eq:A}
A(f) = [\, \bm{a}_1(f), \ldots, \bm{a}_K(f), A_{\z}(f) \, ] \in \mathbb{C}^{M \times M}
\end{align}
is square and generally invertible, and hence the problem can be converted into one that estimates a separation matrix
$W(f) \in \GL(M)$ satisfying $W(f)^h A(f) = I_M$, or equivalently, satisfying
\begin{align}
\label{eq:s=wx}
s_i(f,t) &= \bm{w}_i(f)^h \bm{x}(f,t) \in \mathbb{C}, \quad i = 1,\ldots,K,
\\
\label{eq:z=Wx}
\bm{z}(f,t) &= W_{\z}(f)^h \bm{x}(f,t) \in \mathbb{C}^{M - K},
\end{align}
where we define
\begin{align}
\label{eq:W}
W(f) &= [\, \bm{w}_1(f), \ldots, \bm{w}_K(f), W_{\z}(f) \, ] \in \GL(M),
\\
\bm{w}_i(f) &\in \mathbb{C}^M, \quad i = 1,\ldots,K,
\\
W_{\z}(f) &\in \mathbb{C}^{M \times (M - K)}.
\end{align}

Denote by
$\bm{s}_i(t) = [\, s_i(1,t), \ldots, s_i(F,t) \,]^\top \in \mathbb{C}^F$
the vector of all the frequency components for source $i$ and time-frame $t$.
The proposed IVE exploits the following three assumptions.
Note that Assumption~\ref{assumption:superGaussian} was introduced for developing AuxICA~\cite{ono2010auxica} and
AuxIVA~\cite{ono2011auxiva,ono2012auxiva-stereo,ono2018asj}.
\begin{assumption}[Independence of sources]
The random variables $\{ \bm{s}_i(t), \bm{z}(f,t) \}_{i,f,t}$ are mutually independent:
\begin{align*}
p( \{ \bm{s}_i(t), \bm{z}(f,t) \}_{i,f,t} ) = \prod_{i,t} p(\bm{s}_i(t)) \cdot \prod_{f,t} p(\bm{z}(f,t)).
\end{align*}
\end{assumption}
\begin{assumption}[Super-Gaussian distributions for the target sources~\cite{ono2010auxica,ono2011auxiva}]
\label{assumption:superGaussian}

The target-source signal $\bm{s}_i(t)$ follows a circularly symmetric super-Gaussian distribution:
\begin{align}
- \log p(\bm{s}_i(t)) = G(\| \bm{s}_i(t) \|) + \const,
\end{align}
where $G \colon \mathbb{R}_{\geq 0} \to \mathbb{R}$ is differentiable and satisfies that
$\frac{G'(r)}{r}$ is nonincreasing on $r \in \mathbb{R}_{> 0}$.
Here, $G'$ is the first derivative of $G$.
Candidates of $G$ (or the probability density functions) include
the $\log \cosh$ function
and
the circularly symmetric generalized Gaussian distribution (GGD) with
the scale parameter $\alpha_i \in \mathbb{R}_{> 0}$
and
the shape parameter $0 < \beta < 2$,
which is also known as the exponential power distribution~\cite{gomez1998ggd}:
\begin{align}
G(\| \bm{s}_i(t) \|) = \left( \frac{\| \bm{s}_i(t) \| }{\alpha_i} \right)^\beta
+ 2 F \log \alpha_i.
\end{align}
GGD is a parametric family of symmetric distributions, and when $\beta = 1$ it is nothing but the complex Laplace distribution.
It has been experimentally shown in many studies that ICA type methods including IVA and IVE can work effectively for audio source separation tasks when audio signals such as speech signals are modeled by the super-Gaussian distributions (see, e.g.,~\cite{kim2007,hiroe2006,koldovsky2018ive,koldovsky2017ive,jansky2020adaptive,scheibler2019overiva,scheibler2020fast,scheibler2020ive,ike2020overiva,ono2011auxiva,ono2012auxiva-stereo,ono2018asj}).
\end{assumption}
\begin{assumption}[Stationary Gaussian distribution for the background noise]
\label{assumption:noise}
The noise signal $\bm{z}(f,t) \in \mathbb{C}^{M - K}$ follows a circularly symmetric complex Gaussian distribution with the zero mean and identity covariance matrix:
\begin{align}
\label{eq:z:pdf}
\bm{z}(f,t) &\sim \mathbb{C} \mathcal{N} \left( \bm{0}_{M - K}, I_{M - K} \right),
\\
p(\bm{z}(f,t)) &= \frac{1}{\pi^{M - K}} \exp \left( - \| \bm{z}(f,t) \|^2 \right).
\end{align}
\end{assumption}

Assumption~\ref{assumption:noise} plays a central role for deriving several efficient algorithms for IVE.
Despite this assumption, as we experimentally confirm in Section~\ref{sec:exp},
the proposed IVE can extract speech signals even in a diffuse noise environment where the noise signal is considered super-Gaussian or nonstationary and has an arbitrary large spatial rank.

With the model defined by \eqref{eq:s=wx}--\eqref{eq:z:pdf}, the negative loglikelihood, $g_0(W) \coloneqq - \frac{1}{T} \log p(\bm{x} \mid W)$, can be computed as
\begin{alignat}{2}
\nonumber
g_0(W)
&=&\,& \frac{1}{T} \sum_{i = 1}^K \sum_{t = 1}^T G( \| \bm{s}_i(t) \| )
+ \frac{1}{T} \sum_{f = 1}^F \sum_{t = 1}^T \| \bm{z}(f,t) \|^2
\\
\label{eq:loss}
&&& 
- 2 \sum_{f = 1}^F \log | \det W(f) | + \const,
\end{alignat}
where $W \coloneqq \{ W(f) \}_{f = 1}^F$ are the variables to be optimized.

\begin{remark}
The stationary Gaussian noise model \eqref{eq:z:pdf} with the identity covariance matrix does not sacrifice generality.
At first glance, it seems better to employ
\begin{align}
\label{eq:z:R}
\bm{z}(f,t) \sim \mathbb{C}\mathcal{N} \left( \bm{0}, \Omega_{\z}(f) \right)
\end{align}
with a general covariance matrix $\Omega_{\z}(f) \in \mathcal{S}_{++}^{M - K}$.
However, we can freely change the variables to satisfy \eqref{eq:z:pdf} using the ambiguity between $A_{\z}$ and $\bm{z}$, given by
\begin{align}
A_{\z}(f) \bm{z}(f,t) &= (A_{\z}(f) \Omega_{\z}(f)^{\frac{1}{2}}) (\Omega_{\z}(f)^{- \frac{1}{2}} \bm{z}(f,t)).
\end{align}
\end{remark}

\subsection{Relation to IVA and AuxIVA}
\label{sec:model:IVA}

If we assume that the $M - K$ noise components also independently follow super-Gaussian distributions,
then the IVE model coincides with that of IVA~\cite{kim2007} or AuxIVA~\cite{ono2011auxiva,ono2012auxiva-stereo,ono2018asj}.
The following are the two advantages of assuming the stationary Gaussian model \eqref{eq:z:pdf}.

(i) As we confirm experimentally in Section~\ref{sec:exp}, when we optimize the IVE model, separation filters $\bm{w}_1,\ldots,\bm{w}_K$ extract the top $K$ highly super-Gaussian (or nonstationary) signals such as speech signals from the observed mixture while $W_{\z}$ extracts only the background noises that are more stationary and approximately follow Gaussian distributions.
On the other hand, in IVA, which assumes super-Gaussian noise models, $K$ ($< M$) target-source signals need to be chosen from the $M$ separated signals after optimizing the model.

(ii) As we reveal in Section~\ref{sec:IVE-IP1}, in the proposed IVE,
it suffices to optimize $\image W_{\z} = \{ W_{\z} \bm{v} \mid \bm{v} \in \mathbb{C}^{M - K} \} \subset \mathbb{C}^M$, the subspace spanned by $W_{\z} \in \mathbb{C}^{M \times (M - K)}$, instead of $W_{\z}$.
Because we can optimize $\image W_{\z}$ very efficiently (Section~\ref{sec:IVE-IP1}), IVE can reduce the computation time of IVA, especially when $K \ll M$.

\subsection{Relation to IVE with orthogonal constraint}
\label{sec:model:IVE-OC}

The proposed IVE is inspired by OverIVA with an orthogonal constraint (OC)~\cite{scheibler2019overiva}.
This conventional method will be called IVE-OC in this paper.

IVE-OC~\cite{scheibler2019overiva} was proposed as an acceleration of IVA~\cite{ono2011auxiva} for the case where $K < M$, while maintaining its separation performance.
The IVE-OC model is defined as the proposed IVE by replacing the noise model from \eqref{eq:z:pdf} to \eqref{eq:z:R} and introducing two additional constraints:
\begin{align}
\label{eq:IVE-IP1:Im(Wz):1}
W_{\z}(f) &= \begin{bmatrix}
W_{\z}^1(f) \\
-I_{M - K}
\end{bmatrix}, \quad W_{\z}^1(f) \in \mathbb{C}^{K \times (M - K)},
\\
\label{eq:OC:sample}
& \hspace{-10mm} \frac{1}{T} \sum_{t = 1}^T \bm{s}(f,t) \bm{z}(f,t)^h = O_{K, M - K}.
\end{align}
The first constraint \eqref{eq:IVE-IP1:Im(Wz):1} may be applicable because there is no need to extract the noise components (see~\cite{scheibler2019overiva,scheibler2020ive} for details).
The second constraint \eqref{eq:OC:sample}, called an \textit{orthogonal constraint (OC~\cite{koldovsky2018ive})}, was introduced to help the model distinguish between the target-source and noise signals.

OC, which forces the sample correlation between the separated target-source and noise signals to be zero,
can equivalently be expressed as
\begin{align}
& \hspace{-9 mm} W_{\s}(f)^h V_{\z}(f) W_{\z}(f) = O_{K, M - K},
\\
V_{\z}(f) &= \frac{1}{T} \sum_{t = 1}^T \bm{x}(f,t) \bm{x}(f,t)^h \in \mathcal{S}_{+}^M,
\\
W_{\s}(f) &= [\, \bm{w}_1(f), \ldots,\bm{w}_K(f) \,] \in \mathbb{C}^{M \times K},
\end{align}
which together with \eqref{eq:IVE-IP1:Im(Wz):1} imply
\begin{align}
\label{eq:IVE-IP1:Im(Wz):2}
& \hspace{-3 mm} W_{\z}^1(f) = (W_{\s}(f)^h V_{\z}(f) E_{\s})^{-1} (W_{\s}(f)^h V_{\z}(f) E_{\z}),
\end{align}
where we define
\begin{align}
\label{eq:Es}
E_{\s} &\coloneqq 
[\, \bm{e}_{1}, \ldots, \bm{e}_K \,] = \begin{bmatrix}
I_{K} \\
O_{M - K, K}
\end{bmatrix} \in \mathbb{C}^{M \times K},
\\
\label{eq:Ez}
E_{\z} &\coloneqq
[\, \bm{e}_{K + 1}, \ldots, \bm{e}_M \,] = \begin{bmatrix}
O_{K, M - K} \\
I_{M - K}
\end{bmatrix} \in \mathbb{C}^{M \times (M - K)}.
\end{align}

From \eqref{eq:IVE-IP1:Im(Wz):1} and \eqref{eq:IVE-IP1:Im(Wz):2}, it turns out that $W_{\z}$ is uniquely determined by $W_{\s}$ in the IVE-OC model.
Hence, in a paper on IVE-OC~\cite{scheibler2019overiva}, an algorithm was proposed
in which $W_{\z}$ is updated based on \eqref{eq:IVE-IP1:Im(Wz):1} and \eqref{eq:IVE-IP1:Im(Wz):2} immediately after
updating any other variables $\bm{w}_1,\ldots,\bm{w}_K$ to always impose OC on the model.
Although the algorithm was experimentally shown to work well, 
its validity from a theoretical point of view is unclear
because the update rule for $W_{\z}$ is derived solely from the constraints \eqref{eq:IVE-IP1:Im(Wz):1}--\eqref{eq:OC:sample}
and does not reflect an objective of the optimization problem for parameter estimation, such as minimizing the negative loglikelihood.

In this paper, we develop BCD algorithms for the maximum likelihood estimation of the proposed IVE that does not rely on OC,
and identify one such algorithm (IVE-IP1 developed in Section~\ref{sec:IVE-IP1}) that exactly coincides with the conventional algorithm for IVE-OC.
This means that OC is not essential for developing fast algorithms in IVE-OC.
Due to removing OC from the model, we can provide other more computationally efficient algorithms for IVE,
which is the main contribution of this paper.

\section{Algorithms for the BSE problem}
\label{sec:BSE}

We develop iterative algorithms for the maximum likelihood estimation of IVE.
The proposed and some conventional algorithms~\cite{scheibler2019overiva,scheibler2020fast,scheibler2020ive,ike2020overiva} are based on the following two methods.
\begin{itemize}
\item One is the conventional majorization-minimization (MM) algorithm developed for AuxICA~\cite{ono2010auxica} and AuxIVA~\cite{ono2011auxiva}.
In MM, instead of dealing with original objective function $g_0$ (Eq.~\eqref{eq:loss}),
a surrogate function of $g_0$ that is easier to minimize is addressed (Section~\ref{sec:BSE:MM}).
\item The other is block coordinate descent (BCD) algorithms.
In each iteration of BCDs, several $W$ columns are updated to globally minimize the above surrogate function with respect to that variable.
In this paper, we propose several BCDs that improve the conventional BCDs.
\end{itemize}
Our proposed algorithms are summarized in Algorithm~\ref{alg:main}.
The optimization processes of all the BCDs are summarized in Table~\ref{table:alg}
and detailed in the following subsections.
The computational time complexities of the algorithms are discussed in Section~\ref{sec:computational-complexity}.

\subsection{Majorization-minimization (MM) approach}
\label{sec:BSE:MM}

We briefly describe how to apply the conventional MM technique developed for AuxICA~\cite{ono2010auxica} to the proposed IVE.

In IVE as well as AuxICA/AuxIVA, a surrogate function $g$ of $g_0$ is designed with an auxiliary variable $r$ that satisfies
\begin{align}
g_0(W) = \min_{r} g (W, r).
\end{align}
Then, variables $W$ and $r$ are alternately updated by iteratively solving
\begin{align}
\label{problem:MM:1}
r^{(l)} &\in \argmin_{r} g(W^{(l - 1)}, r),
\\
\label{problem:MM:2}
W^{(l)} &\in \argmin_{W} g(W, r^{(l)})
\end{align}
for $l = 1,2,\ldots$ until convergence.
In the same way as in AuxIVA~\cite{ono2011auxiva},
by applying Proposition~\ref{prop:MM} in Appendix~\ref{appendix:lemma} to the first term of $g_0(W)$,
problem \eqref{problem:MM:2} in IVE comes down to solving the following $F$ subproblems:
\begin{align}
\label{problem:maxdet}
W(f) \in \argmin_{W(f)} ~ g (W(f), r^{(l)}), \quad f = 1,\ldots,F,
\end{align}
where $r^{(l)} \coloneqq \{ r_i^{(l)}(t) \in \mathbb{R}_{\geq 0} \}_{i,t}$ and
\begin{align}
\nonumber
&\hspace{-10 mm}
g( W(f), r^{(l)} ) = \sum_{i = 1}^K \bm{w}_i(f)^h V_i(f) \bm{w}_i(f)
\\
\label{eq:loss:MM}
&\hspace{-5 mm}
+ \trace \left( W_{\z}(f)^h V_{\z}(f) W_{\z}(f) \right) 
- 2\log | \det W(f) |,
\\
\label{eq:Vi}
V_i(f) &= \frac{1}{T} \sum_{t = 1}^T \phi_{i}(t) \bm{x}(f,t) \bm{x}(f,t)^h,
\\
\label{eq:Vz}
V_{\z}(f) &= \frac{1}{T} \sum_{t = 1}^T \bm{x}(f,t) \bm{x}(f,t)^h,
\\
\label{eq:MM:phi}
\phi_{i}(t) &= \frac{ G'(r_i^{(l)}(t)) }{ 2 r_i^{(l)}(t) }, \quad r_i^{(l)}(t) = \| \bm{s}_i^{(l)}(t) \|,
\\
\label{eq:MM:si}
&\hspace{-7 mm}
s_i^{(l)}(f,t) = \bm{w}_i^{(l - 1)}(f)^h \bm{x}(f,t).
\end{align}
Here, the computation of \eqref{eq:MM:phi}--\eqref{eq:MM:si} corresponds to the optimization of \eqref{problem:MM:1}.
Recall that $G'$ in \eqref{eq:MM:phi} is the first derivative of $G \colon \mathbb{R}_{\geq 0} \to \mathbb{R}$ (see Assumption~\ref{assumption:superGaussian}).
To efficiently solve \eqref{problem:maxdet}, we propose BCD algorithms in the following subsections.
From the derivation, it is guaranteed that the objective function is monotonically nonincreasing at each iteration in the MM algorithm.%
\footnote{
Due to space limitations, 
showing the convergence rate and other convergence properties of the proposed algorithms will be left as a future work.
}

\subsection{Block coordinate descent (BCD) algorithms and the stationary condition for problem \eqref{problem:maxdet}}
\label{sec:BSE:KKT}

No algorithms have been found that obtain a global optimal solution for problem~\eqref{problem:maxdet} for general $K, M \in \mathbb{N}$.
Thus, iterative algorithms have been proposed to find a local optimal solution.
Among them, a family of BCD algorithms has been attracting much attention 
(e.g.,~\cite{ono2010auxica,ono2011auxiva,ono2012auxiva-stereo,ono2018asj,scheibler2019overiva,scheibler2020ive,scheibler2020fast,ike2020overiva,kitamura2016ilrma,kameoka2019MVAE,makishima2019independent,sekiguchi2019fast})
because they have been experimentally shown to work faster and more robustly than other algorithms
such as the natural gradient method~\cite{amari1996natural-gradient}.
The family of these BCD algorithms (specialized to solve \eqref{problem:maxdet}) is currently called an \textit{iterative projection (IP) method}.

As we will see in the following subsections, all the IP algorithms summarized in Table~\ref{table:alg}
can be developed by exploiting 
the stationary condition, which is
also called the first-order necessary optimality condition~\cite{nocedal-Jorge2006optimization}.
To simplify the notation, when we discuss~\eqref{problem:maxdet}, we abbreviate frequency bin index $f$ without mentioning it.
For instance, $W(f)$ and $V_i(f)$ are simply denoted as $W$ and $V_i$ (without confusion).
\begin{lemma}[See, e.g., \cite{pham2001,degerine2004,degerine2006maxdet}]
\label{lemma:KKT}
The stationary condition for problem~\eqref{problem:maxdet} is expressed as ($i = 1,\ldots,K$)
\begin{alignat}{3}
\label{eq:KKT:wk}
\frac{\partial g}{\partial \bm{w}_i^\ast} &= \bm{0}_M &\quad &\Longleftrightarrow &\quad W^h V_i \bm{w}_i &= \bm{e}_i,
\\
\label{eq:KKT:Wz}
\frac{\partial g}{\partial W_{\z}^\ast} &= O_{M, M - K} &\quad &\Longleftrightarrow &\quad W^h V_{\z} W_{\z} &= E_{\z},
\end{alignat}
where $E_{\z}$ is given by \eqref{eq:Ez}.
\end{lemma}

To rigorously derive the algorithms, we always assume the following two technical but mild conditions (C1) and (C2) for problem \eqref{problem:maxdet}:%
\footnote{
If (C1) is violated, problem \eqref{problem:maxdet} has no optimal solutions and algorithms should diverge to infinity (see {\cite[Proposition 1]{ike2019ilrma}} for the proof).
Conversely, if (C1) is satisfied, it is guaranteed that problem \eqref{problem:maxdet} has an optimal solution by Proposition~\ref{prop:loss:lower-bounded} in Appendix~\ref{appendix:lemma}.
In practice, the number of frames $T$ exceeds that of sensors $M$, and (C1) holds in general.

Condition (C2) is satisfied automatically if we initialize $W$ as nonsingular.
Intuitively, singular $W$ implies $- \log | \det W | = +\infty$, which will never occur during optimization.
}
\begin{description}
\item[(C1)] $V_1, \ldots, V_K, V_{\z} \in \mathcal{S}_{+}^M$ are positive definite.
\item[(C2)] Estimates of $W \in \mathbb{C}^{M \times M}$ are always nonsingular during optimization.
\end{description}

\subsection{Conventional methods: IVA-IP1, IVA-IP2, and IVE-OC}
\label{sec:BSS-IP}

\subsubsection{IVA-IP1}
\label{sec:BSS-IP1}

Let $W_{\z} = [\bm{w}_{K + 1}, \ldots, \bm{w}_M]$.
As shown in Table~\ref{table:alg}, IVA-IP1~\cite{ono2011auxiva}
cyclically updates each separation filter $\bm{w}_1,\ldots,\bm{w}_M$ by solving
the following subproblem for each $i = 1,\ldots,M$ one by one:
\begin{align}
\label{problem:BSS-IP1}
\bm{w}_i \in \argmin_{\bm{w}_i} g (\bm{w}_1,\ldots,\bm{w}_M, r).
\end{align}
This can be solved under (C1) and (C2) by
\begin{align}
\label{eq:BSS-IP1:1}
\bm{u}_i &\leftarrow (W^h V_i)^{-1} \bm{e}_i \in \mathbb{C}^M,
\\
\label{eq:BSS-IP1:2}
\bm{w}_i &\leftarrow \bm{u}_i \left( \bm{u}_i^h V_i \bm{u}_i \right)^{-\frac{1}{2}} \in \mathbb{C}^M.
\end{align}
Here, we define $V_i \coloneqq V_{\z} \in \mathcal{S}_{++}^M$ for $i = K + 1, \ldots, M$, and the $i$th column of $W$ in \eqref{eq:BSS-IP1:1}, i.e., $\bm{w}_i$, is set to the current value before update.
When applied to the BSE problem, IVA-IP1's main drawback is that
the computation time increases significantly as $M$ gets larger since it updates $W_{\z}$ even though there is no need to extract the background noises.

\subsubsection{IVE-OC}
\label{sec:IVE-IP1:OC}

To accelerate IVA-IP1 when it is applied to the BSE problem, IVE-OC~\cite{scheibler2019overiva} updates
$W_{\z}$ using \eqref{eq:IVE-IP1:Im(Wz):1} and \eqref{eq:IVE-IP1:Im(Wz):2}
(see Section~\ref{sec:model:IVE-OC} and Algorithm~\ref{alg:IVE-IP1}).
Although this update rule seems heuristic, we reveal in Section~\ref{sec:IVE-IP1}
that IVE-OC can be comprehended in terms of BCD for the proposed IVE that does not rely on OC.

\subsubsection{IVA-IP2}
\label{sec:IVA-IP2}

When $(K, M) = (1, 2)$ or $(2, 2)$, i.e., when $W = [\bm{w}_1, \bm{w}_2]$, problem \eqref{problem:maxdet} can be solved directly (not iteratively) through a generalized eigenvalue problem~\cite{degerine2006maxdet,ono2012auxiva-stereo},
which is more efficient than IVA-IP1.
We extend this direct method to the case where $K = 1$ and general $M$ ($\geq 2$) in Section~\ref{sec:IVE-IP2:K=1}.

In a previous work~\cite{ono2018asj},
this IVA-IP2 was extended for the case where $K = M \geq 3$.
This algorithm, which is also called IVA-IP2, updates two separation filters, e.g., $\bm{w}_i$ and $\bm{w}_j$, in each iteration by exactly solving the following subproblem:
\begin{align}
(\bm{w}_i, \bm{w}_j) \in \argmin_{\bm{w}_i,\,\bm{w}_j} g (\bm{w}_1,\ldots,\bm{w}_M, r).
\end{align}
We extend this algorithm in Section~\ref{sec:IVE-IP2-new:K>1}.

\subsection{Proposed IVE-IP2 for the case of $K = 1$}
\label{sec:IVE-IP2:K=1}

We focus on the case where $K = 1$ and derive the proposed algorithm IVE-IP2%
\footnote{
The algorithm IVE-IP2 for $K = 1$ was developed independently and simultaneously by Scheibler--Ono (called the fast IVE or FIVE~\cite{scheibler2020fast}) and the authors~\cite{ike2020overiva} in the Proceedings of ICASSP2020.
In this paper, we give a complete proof for the derivation of IVE-IP2 and add some remarks
(see also Section~\ref{sec:source-extraction} for the projection back operation).
}
that is expected to be more efficient than IVE-IP1.
The algorithm is summarized in Algorithm~\ref{alg:IVE-IP2:K=1}.

When $(K, M) = (1, 2)$ or $(2, 2)$, problem \eqref{problem:maxdet} can be solved directly through a generalized eigenvalue problem~\cite{degerine2006maxdet,ono2012auxiva-stereo,ono2018asj}.
We here extend this direct method to the case where $K = 1$ and $M \geq 2$ in the following proposition.
\begin{proposition}
\label{prop:IVE-IP2:K=1}
Let $K = 1$, $M \geq 2$, and $V_1,V_{\z} \in \mathcal{S}_{++}^M$.
A matrix $W = [ \bm{w}_1, W_{\z} ] \in \GL(M)$ with $\bm{w}_1 \in \mathbb{C}^M$ satisfies
stationary conditions \eqref{eq:KKT:wk} and \eqref{eq:KKT:Wz} if and only if
\begin{align}
\label{eq:IVE-IP2:K=1:w1}
\bm{w}_1 &= \bm{u}_1 \left( \bm{u}_1^h V_1 \bm{u}_1  \right)^{- \frac{1}{2}} Q_1, \quad Q_1 \in \mathcal{U}(1),
\\
\label{eq:IVE-IP2:K=1:Wz}
W_{\z} &= U_{\z} \left( U_{\z}^h V_{\z} U_{\z} \right)^{- \frac{1}{2}} Q_{\z}, \quad Q_{\z} \in \mathcal{U}(M - 1),
\\
\label{eq:IVE-IP2:K=1:orth}
U_{\z} &\in \mathbb{C}^{M \times (M - 1)} \quad \text{with} \quad U_{\z}^h V_{\z} \bm{u}_1 = \bm{0}_{M - 1},
\\
\label{eq:IVE-IP2:K=1:eig}
\bm{u}_1 &\in \mathbb{C}^{M \times 1} \quad \text{with} \quad V_{\z} \bm{u}_1 = \lambda V_1 \bm{u}_1, 
\end{align}
where $\mathcal{U}(d)$ is the set of all unitary matrices of size $d \times d$.
Also, \eqref{eq:IVE-IP2:K=1:eig} is the generalized eigenvalue problem for $(V_{\z}, V_1)$
with the eigenvalue $\lambda \in \mathbb{R}_{> 0}$ and eigenvector $\bm{u}_1 \in \mathbb{C}^M$.

Moreover, if the generalized eigenvalue $\lambda \in \mathbb{R}_{> 0}$ in \eqref{eq:IVE-IP2:K=1:eig} is chosen as the largest one,
then any $W \in \GL(M)$ obtained by \eqref{eq:IVE-IP2:K=1:w1}--\eqref{eq:IVE-IP2:K=1:eig} 
is a global optimal solution for problem \eqref{problem:maxdet}.
\end{proposition}
\begin{proof}
We first show that $W$ satisfies \eqref{eq:KKT:wk}--\eqref{eq:KKT:Wz} if and only if it is computed by \eqref{eq:IVE-IP2:K=1:w1}--\eqref{eq:IVE-IP2:K=1:eig}.

For the ``if'' part, observe that
$U_{\z} V_1 \bm{u}_1 = \bm{0}_{M - 1}$ holds by \eqref{eq:IVE-IP2:K=1:orth}--\eqref{eq:IVE-IP2:K=1:eig} and $\lambda \neq 0$.
Thus, $W$ surely satisfies \eqref{eq:KKT:wk}--\eqref{eq:KKT:Wz}.

We prove the ``only if'' part.
The stationary conditions~\eqref{eq:KKT:wk}--\eqref{eq:KKT:Wz} imply that vectors $V_1 \bm{w}_1$ and $V_{\z} \bm{w}_1$ are orthogonal to the subspace $\image W_{\z}$ of dimension $M - 1$.
Hence, it holds that $\bm{w}_1= c \bm{u}_1$ for some $c \in \mathbb{C}$, where $\bm{u}_1$ is given by \eqref{eq:IVE-IP2:K=1:eig}.
This $c$ is restricted by $\bm{w}_1^h V_1 \bm{w}_1 = 1$, and we obtain \eqref{eq:IVE-IP2:K=1:w1}.
In a similar manner, \eqref{eq:KKT:Wz} implies that vector $V_{\z} \bm{u}_1$ is orthogonal to $\image W_{\z}$ of dimension $M - 1$.
Hence, it holds that $W_{\z} = U_{\z} R$ for some $R \in \mathbb{C}^{(M - 1) \times (M - 1)}$, where $U_{\z}$ is given by \eqref{eq:IVE-IP2:K=1:orth}.
This $R$ is restricted by $W_{\z}^h V_{\z} W_{\z} = I_{M - 1}$, and we have \eqref{eq:IVE-IP2:K=1:Wz}.

We next show the latter statement.
By Proposition~\ref{prop:loss:lower-bounded} in Section~\ref{sec:IVE-IP1}, global optimal solutions exist, and they must satisfy the stationary conditions, which are equivalent to \eqref{eq:IVE-IP2:K=1:w1}--\eqref{eq:IVE-IP2:K=1:eig}.
Since \eqref{eq:IVE-IP2:K=1:w1}--\eqref{eq:IVE-IP2:K=1:eig} satisfy \eqref{eq:KKT:wk}--\eqref{eq:KKT:Wz},
the sum of the first and second terms of $g$ becomes $M$, which is constant.
On the other hand, for the $\log \det$ term, it holds that
\begin{align*}
| \det W |&= \left( \bm{u}_1^h V_1 \bm{u}_1 \right)^{-\frac{1}{2}} \cdot \det \left( U_{\z}^h V_{\z} U_{\z} \right)^{-\frac{1}{2}} \cdot | \det U |
\\
&= \sqrt{\lambda} \det \left( U^h V_{\z} U \right)^{-\frac{1}{2}} \cdot | \det U |
= \sqrt{\lambda} \det (V_{\z})^{-\frac{1}{2}},
\end{align*}
where we define $U \coloneqq [\bm{u}_1, U_{\z}]$ and use $U_{\z}^h V_{\z} \bm{u}_1 = \bm{0}_{M - 1}$ in the second equality.
Hence, the largest $\lambda$ leads to the smallest $g$, which concludes the proof.
\end{proof}

By Proposition~\ref{prop:IVE-IP2:K=1}, under condition (C1), a global optimal solution for problem~\eqref{problem:maxdet}
can be obtained by updating $W = [\bm{w}_1, W_{\z}]$ using \eqref{eq:IVE-IP2:K=1:w1}--\eqref{eq:IVE-IP2:K=1:eig} with $Q_1 = 1$ and $Q_{\z} = I_{M - 1}$
and choosing the generalized eigenvalue $\lambda$ in \eqref{eq:IVE-IP2:K=1:eig} as the largest one.
Moreover, this algorithm can be accelerated by omitting the computation of \eqref{eq:IVE-IP2:K=1:Wz}--\eqref{eq:IVE-IP2:K=1:orth} and updating only $\bm{w}_1$ according to \eqref{eq:IVE-IP2:K=1:w1} and \eqref{eq:IVE-IP2:K=1:eig}.
It is applicable for extracting the unique target-source signal because the formulas for computing $\bm{w}_1$, $V_1$, and $V_{\z}$ are independent of $W_{\z}$.
The obtained algorithm IVE-IP2 is shown in Algorithm~\ref{alg:IVE-IP2:K=1}.

Interestingly, because $V_{\z}$ and $V_1$ can be viewed as the covariance matrices of the mixture and noise signals,
the update rules \eqref{eq:IVE-IP2:K=1:w1} and \eqref{eq:IVE-IP2:K=1:eig} turn out to be a MaxSNR beamformer~\cite{vantrees2004,warsitz2007maxsnr}.
Hence, IVE-IP2 can be understood as a method that iteratively updates the MaxSNR beamformer $\bm{w}_1$
and target-source signal $\bm{s}_1$.

\subsection{Drawback of conventional IVE-IP2 for the case of $K \geq 2$}
\label{sec:IVE-IP2-old:K>1}

Suppose $2 \leq K < M$.
The conventional IVE-IP2 developed in {\cite[Algorithm 2]{scheibler2020ive}}
is a BCD algorithm that cyclically updates the pairs $(\bm{w}_1, W_{\z}), \ldots, (\bm{w}_K, W_{\z})$ by solving the following subproblem for each $i = 1,\ldots, K$ one by one:
\begin{align}
\label{problem:IVE-IP2}
(\bm{w}_i, W_{\z}) \in \argmin_{ \bm{w}_i,\, W_{\z} } g (\bm{w}_1, \ldots, \bm{w}_K, W_{\z}, r).
\end{align}
It was shown in \cite[Theorem~3]{scheibler2020ive} that a global optimal solution of problem \eqref{problem:IVE-IP2} can be obtained through a generalized eigenvalue problem.
In this paper, we simplify the result of \cite[Theorem~3]{scheibler2020ive} in the next proposition,
based on which we will speed up the conventional IVE-IP2 in Section~\ref{sec:IVE-IP2-new:K>1}.
\begin{proposition}
\label{prop:IVE-IP2:K>1}
Let $2 \leq K < M$, and $V_i,V_{\z} \in \mathcal{S}_{++}^M$.
A global optimal solution of problem~\eqref{problem:IVE-IP2} is obtained as
\begin{align}
\label{eq:IVE-IP2:K>1:wi}
\bm{w}_i &=  P_i \bm{b}_i \left( \bm{b}_i^h G_i \bm{b}_i  \right)^{- \frac{1}{2}} \in \mathbb{C}^{M \times 1},
\\
\label{eq:IVE-IP2:K>1:Wz}
W_{\z} &= P_{\z} B_{\z} \left( B_{\z}^h G_{\z} B_{\z} \right)^{- \frac{1}{2}} \in \mathbb{C}^{M \times (M - K)},
\\
\label{eq:IVE-IP2:K>1:P}
P_\ell &= \left((W')^h V_\ell \right)^{-1} [\, \bm{e}_i, E_{\z} \,] \in \mathbb{C}^{M \times (M - K +1)},
\\
\label{eq:IVE-IP2:K>1:G}
G_\ell &= P_\ell^h V_\ell P_\ell \in \mathcal{S}_{++}^{M - K + 1}, \quad \ell \in \{ i, \z \},
\\
\label{eq:IVE-IP2:K>1:W'}
W' &= [\, \bm{w}_1, \ldots, \bm{w}_i', \ldots, \bm{w}_K, W'_{\z} \,] \in \GL(M),
\\
\label{eq:IVE-IP2:K>1:orth}
B_{\z} &\in \mathbb{C}^{(M - K + 1) \times (M - K)} \quad \text{with} \quad B_{\z}^h G_{\z} \bm{b}_i = \bm{0}_{M - K},
\\
\label{eq:IVE-IP2:K>1:eig}
\bm{b}_i &\in \mathbb{C}^{(M - K + 1) \times 1} \quad \text{with} \quad G_i \bm{b}_i = \lambda_{\max} G_{\z} \bm{b}_i, 
\end{align}
where $E_{\z}$ is defined by \eqref{eq:Ez},
and $\bm{w}_i' \in \mathbb{C}^{M \times 1}$ and $W'_{\z} \in \mathbb{C}^{M \times (M - K)}$ in \eqref{eq:IVE-IP2:K>1:W'} are set arbitrarily as long as $W'$ is nonsingular (for instance they can be set to the current values under condition (C2)).
Also, \eqref{eq:IVE-IP2:K>1:eig} is the generalized eigenvalue problem for $(G_i,G_{\z})$,
and $\bm{b}_i$ is the eigenvector corresponding to the largest generalized eigenvalue $\lambda_{\max} \in \mathbb{R}_{> 0}$.
\end{proposition}
\begin{proof}
The proof is given in Appendix~\ref{appendix:proof-of-prop3}.
\end{proof}
In the conventional IVE-IP2, under condition (C2),
problem \eqref{problem:IVE-IP2} is solved by \eqref{eq:IVE-IP2:K>1:wi}--\eqref{eq:IVE-IP2:K>1:W'}
where $B_{i,\z} \coloneqq [\bm{b}_i, B_{\z}]$ is computed through the following generalized eigenvalue problem
instead of using \eqref{eq:IVE-IP2:K>1:orth}--\eqref{eq:IVE-IP2:K>1:eig}:
\begin{align}
\label{eq:IVE-IP2:K>1:old:GEV}
G_i B_{i,\z} = G_{\z} B_{i,\z} \diag \{ \lambda_\mathrm{max}, \lambda_2, \ldots,\lambda_{M - K + 1} \}.
\end{align}
Since the computational cost of solving \eqref{eq:IVE-IP2:K>1:old:GEV} is greater than computing \eqref{eq:IVE-IP2:K>1:orth}--\eqref{eq:IVE-IP2:K>1:eig},
we can speed up the conventional IVE-IP2, as shown in the next subsection.

\subsection{Proposed IVE-IP2 for $K \geq 2$}
\label{sec:IVE-IP2-new:K>1}

The proposed IVE-IP2 for $K \geq 2$, which is summarized in Algorithm~\ref{alg:IVE-IP2:K>1}, is an acceleration of the conventional IVE-IP2 described in the previous subsection.
We here provide another efficient formula to solve \eqref{problem:IVE-IP2}
without changing the behavior of the conventional IVE-IP2.

Thanks to newly providing Proposition~\ref{prop:IVE-IP2:K>1},
we can update the pair $(\bm{w}_i, W_{\z})$ according to \eqref{eq:IVE-IP2:K>1:wi}--\eqref{eq:IVE-IP2:K>1:eig} in which
only the first generalized eigenvector has to be computed.
Moreover, for the purpose of optimizing $\bm{w}_1,\ldots,\bm{w}_K$,
we do not have to update $W_{\z}$
using Eqs. \eqref{eq:IVE-IP2:K>1:Wz} and \eqref{eq:IVE-IP2:K>1:orth}
when solving problem \eqref{problem:IVE-IP2} for each $i = 1,\ldots,K$.
To see this, observe that 
\begin{enumerate}
\item in the MM algorithm (Eqs. \eqref{problem:maxdet}--\eqref{eq:MM:si}), the auxiliary variable $r$, and by extension, the covariance matrices $V_1, \ldots, V_K, V_{\z}$ are independent of $W_{\z}$; and hence,
\item $W_{\z}$ never contributes to the construction of the surrogate function $g$ during iterative optimization.
\end{enumerate}
This observation implies that problem \eqref{problem:IVE-IP2} remains the same during iterative optimization regardless whether we update $W_{\z}$ or not.
Hence, in IVE-IP2, it is sufficient to update only $\bm{w}_i$ in \eqref{problem:IVE-IP2} using Eqs. \eqref{eq:IVE-IP2:K>1:wi}, \eqref{eq:IVE-IP2:K>1:P}--\eqref{eq:IVE-IP2:K>1:W'},
\eqref{eq:IVE-IP2:K>1:eig}.

This advantage stems from the stationary Gaussian assumption for the noise components.
To the contrary, suppose that the noise components are nonstationary or non-Gaussian.
Then, the surrogate function $g$ as well as $V_{\z}$ should depend on $W_{\z}$ in the same way that $V_i$ depends on $\bm{w}_i$, meaning that $W_{\z}$ must be optimized in subproblem~\eqref{problem:IVE-IP2}.
In this way, the stationary Gaussian assumption is important to derive a computationally efficient algorithm for IVE.

\subsection{IVE-IP1: Reinterpretation of the conventional IVE-OC}
\label{sec:IVE-IP1}

In this subsection, we model the noise components by \eqref{eq:z:R} with a general time-independent covariance matrix $\Omega_{\z}$, i.e.,
\begin{align}
- \log p(\bm{z}(t)) = \bm{z}(t)^h \Omega_{\z}^{-1} \bm{z}(t) + \log \det \Omega_{\z} + \const,
\end{align}
and 
develop a BCD algorithm called IVE-IP1 that cyclically updates $\bm{w}_1,(W_{\z},\Omega_{\z}),\bm{w}_2,(W_{\z},\Omega_{\z}), \ldots,\bm{w}_K, (W_{\z},\Omega_{\z})$ by solving the following subproblems one by one:
\begin{align}
\label{problem:IVE-IP1:wk}
\bm{w}_i &\in \argmin_{\bm{w}_i} g (\bm{w}_1,\ldots,\bm{w}_K,W_{\z},\Omega_{\z},r),
\\
\label{problem:IVE-IP1:Wz}
(W_{\z},\Omega_{\z}) &\in \argmin_{W_{\z},\, \Omega_{\z}} g (\bm{w}_1,\ldots,\bm{w}_K,W_{\z},\Omega_{\z},r).
\end{align}
Here, the frequency bin index $f$ is omitted.
Due to considering the noise covariance $\Omega_{\z}$, the objective function $g$ is slightly changed to
\begin{alignat}{2}
\nonumber
g( W, \Omega_{\z}, r) &=&\,& \sum_{i = 1}^K \bm{w}_i^h V_i \bm{w}_i 
+ \trace \left( W_{\z}^h V_{\z} W_{\z} \Omega_{\z}^{-1} \right)
\\
\label{eq:obj:ip1}
&&&+ \log \det \Omega_{\z} - 2\log | \det W |.
\end{alignat}
It will turn out that the BCD algorithm IVE-IP1 is identical to the conventional IVE-OC summarized in Algorithm~\ref{alg:IVE-IP1}.
In other words, IVE-OC, which relies on the orthogonal constraint (OC), can be interpreted as a BCD algorithm for the proposed IVE without OC.

Since problem \eqref{problem:IVE-IP1:wk} is the same as problem \eqref{problem:BSS-IP1},
the update rules for $\bm{w}_1, \ldots, \bm{w}_K$ are given by \eqref{eq:BSS-IP1:1}--\eqref{eq:BSS-IP1:2}.
On the other hand, an algorithm for problem \eqref{problem:IVE-IP1:Wz} can be obtained from the next proposition.
\begin{proposition}
\label{prop:IVE-IP1:Wz}
Let $K, M \in \mathbb{N}$ with $K < M$, 
$V_{\z} \in \mathcal{S}_{++}^M$,
and let $W_{\s} = [\bm{w}_1,\ldots,\bm{w}_K] \in \mathbb{C}^{M \times K}$ be full column rank.
Then, a pair $(W_{\z}, \Omega_{\z})$ is a global optimal solution of problem \eqref{problem:IVE-IP1:Wz} 
if and only if it is computed by
\begin{align}
\label{eq:prop:Wz}
W_{\z} &\in \mathbb{C}^{M \times (M - K)} \quad \text{with} \quad W_{\s}^h V_{\z} W_{\z} = O_{K, M - K},
\\
\label{eq:prop:Wz:2}
\Omega_{\z} &= W_{\z}^h V_{\z} W_{\z} \in \mathcal{S}_{++}^{M - K}.
\end{align}
(Note that $W_{\z}$ must be full column rank to guarantee the positive definiteness of $\Omega_{\z}$.)
\end{proposition}
\begin{proof}
The stationary condition which is the necessary optimality condition of problem \eqref{problem:IVE-IP1:Wz} is expressed as
\begin{alignat}{3}
\label{eq:prop-ip1:1}
\frac{\partial g}{\partial W_{\z}^\ast} &= O &~~
&\Longleftrightarrow&~~&
\begin{cases}
W_{\s}^h V_{\z} W_{\z} = O_{K, M - K},
\\
W_{\z}^h V_{\z} W_{\z} = \Omega_{\z},
\end{cases}
\\
\label{eq:prop-ip1:3}
\frac{\partial g}{\partial \Omega_{\z}^{-1}} &= O &~~
&\Longleftrightarrow&~~&
W_{\z}^h V_{\z} W_{\z} = \Omega_{\z}.
\end{alignat}
Hence, the ``only if'' part is obvious.

To see the ``if'' part, we show that all stationary points are globally optimal.
To this end, it suffices to prove that
\begin{description}
\item[(i)] $g$ takes the same value on all stationary points; and
\item[(ii)] $g$ attains its minimum at some $(W_{\z}, \Omega_{\z})$.
\end{description}

We first check (i).  
By~\eqref{eq:prop:Wz:2}, the second term of $g$ becomes
$\trace (W_{\z}^h V_{\z} W_{\z} \Omega_{\z}^{-1} ) = M - K$, which is constant.
Let $(W_{\z}, \Omega_{\z})$ and $(W_{\z}', \Omega_{\z}')$ be two stationary points.
Since $\image W_{\z} = \image W'_{\z}$ by the first part of \eqref{eq:prop-ip1:1}, we have $W_{\z}' = W_{\z} Q$ for some $Q \in \GL(M - K)$.
Then, by \eqref{eq:prop-ip1:3}, $\Omega_{\z}' = Q^h \Omega_{\z} Q$.
Hence, for the $\log \det$ terms of $g$, we have
\begin{alignat*}{2}
&&\,& \log \det \Omega_{\z}' - 2\log |\det [W_{\s}, W_{\z}'] |
\\
&=&& \log \det (Q^h \Omega_{\z} Q) - 2 \log |\det [W_{\s}, W_{\z}] \cdot \det Q|
\\
&=&& \log \det \Omega_{\z} - 2\log |\det [W_{\s}, W_{\z}] |,
\end{alignat*}
which concludes (i).

We next show (ii).
Let us change the variables from $W_{\z}$ and $\Omega_{\z}$ to $W_{\z}' = W_{\z} \Omega_{\z}^{- \frac{1}{2}}$ and $\Omega_{\z}$.
Then, the objective function $g$ with respect to $W_{\z}'$ and $\Omega_{\z}$ is expressed as
\begin{align*}
g(W_{\z}', \Omega_{\z}) =
\trace (W_{\z}'^h V_{\z} W_{\z}') - 2\log | \det [W_{\s}, W_{\z}' ] | + \const,
\end{align*}
which is independent of $\Omega_{\z}$.
By Proposition~\ref{prop:loss:lower-bounded},
the problem of minimizing $g$ with respect to $W_{\z}'$ attains its minimum at some $W_{\z}'$, which in turn implies that problem \eqref{problem:IVE-IP1:Wz} also attains its minimum at some $(W_{\z}, \Omega_{\z})$.
\end{proof}
Proposition~\ref{prop:IVE-IP1:Wz} implies that 
in problem \eqref{problem:IVE-IP1:Wz}
it is sufficient to optimize $\image (W_{\z})$, instead of $W_{\z}$, to satisfy the orthogonal constraint (OC), i.e., Eq. \eqref{eq:prop:Wz},
which is part of the stationary condition \eqref{eq:prop-ip1:1}.%
\footnote{
The observation that OC appears as part of the stationary condition was pointed out in \cite{scheibler2020ive}.
}
Since the update formula \eqref{eq:IVE-IP1:Im(Wz):1} and \eqref{eq:IVE-IP1:Im(Wz):2} for $W_{\z}$ in IVE-OC surely satisfies OC,
it is applicable for the solution of problem \eqref{problem:IVE-IP1:Wz}.
Hence, it turns out that the conventional IVE-OC summarized in Algorithm~\ref{alg:IVE-IP1} can be obtained as BCD for the proposed IVE.

\subsection{Source extraction and projection back}
\label{sec:source-extraction}

After optimizing separation matrix $W$ in IVE, we can achieve source extraction by computing the minimum mean square error (MMSE) estimator of the source spatial image $\bm{x}_i$ for each $i = 1,\ldots,K$:
\begin{align}
\label{eq:MMSE}
\hat{\bm{x}}_i(f,t) \coloneqq
\underbrace{ \left( W(f)^{-h} \bm{e}_i \right) }_{ \text{projection back} } 
\underbrace{ \left( \bm{w}_i(f)^h \bm{x}(f,t) \right) }_{ \text{source extraction} }.
\end{align}
The projection back operation~\cite{murata2001projection-back} adjusts the amplitude ambiguity of the extracted signals between frequency bins.

We here show that in the projection back operation, i.e., $W^{-h} \bm{e}_i$ for $i = 1,\ldots,K$,
we need only $\bm{w}_1,\ldots,\bm{w}_K, \image W_{\z}$ but not $W_{\z}$
(frequency bin index $f$ is dropped off for simplicity).
To see this, let 
\begin{align}
\label{eq:W^i}
W^{(i)} = [\bm{w}_1,\ldots,\bm{w}_K, W_{\z}^{(i)}] \in \GL(M), \quad i \in \{1, 2 \}
\end{align}
satisfy $\image W_{\z}^{(1)} = \image W_{\z}^{(2)}$.
Then, we have
\begin{align*}
W^{(1)} = W^{(2)} \begin{bmatrix}
I_K & O \\
O & D \\
\end{bmatrix}
~\text{for some $D \in \GL(M - K)$}.
\end{align*}
This implies $\left( W^{(1)} \right)^{-h} \bm{e}_i = \left( W^{(2)} \right)^{-h} \bm{e}_i$ for $i = 1,\ldots,K$,
meaning that the projection back operation depends only on $\image W_{\z}$, as we required.

\begin{remark}
\label{remark:projection-back}
Since the proposed IVE-IP2 will never update $W_{\z}$ from its initial value,
we update $\image W_{\z}$ using \eqref{eq:IVE-IP1:Im(Wz):1} and \eqref{eq:IVE-IP1:Im(Wz):2} after the optimization,
which corresponds to Step~\ref{step:Im(Wz)} in Algorithm~\ref{alg:main}.
This update rule is applicable for this purpose
by considering the stationary condition~\eqref{eq:KKT:Wz}, 
which is equivalent to~\eqref{eq:prop-ip1:1} with $\Omega_{\z} = I_{M - K}$.
\end{remark}

\section{Algorithm for semiblind BSE problem}
\label{sec:semi-BSE}

\subsection{Proposed Semi-IVE}
\label{sec:semi-ive}

We next address the semiblind source extraction problem (Semi-BSE) defined in Section~\ref{sec:problem}.
In Semi-BSE, we are given a priori the $L$ transfer functions (steering vectors) for target sources $l = 1,\ldots,L$ ($1 \leq L \leq K$), denoted as
\begin{align}
A_1(f) \coloneqq [\, \bm{a}_1(f), \ldots, \bm{a}_L(f) \,] \in \mathbb{C}^{M \times L}.
\end{align}
In this situation, it is natural in the context of the linear constrained minimum variance (LCMV~\cite{vantrees2004}) beamforming algorithm
to regularize the separation matrix as 
$W(f)^h A_1(f) = E_1$, where $E_1$ is defined by \eqref{eq:E1} below.
We can thus immediately propose the Semi-BSE method called Semi-IVE
as IVE with the linear constraints, i.e.,
\begin{align}
\label{problem:semi-BSE:main}
\left.
\begin{array}{cl}
\underset{W}{\minimize} & g_0 (W) \quad \text{(defined by \eqref{eq:loss})} \\
\subject-to & W(f)^h A_1(f) = E_1, \quad f = 1,\ldots,F.
\end{array}
\right\}
\end{align}
Here, $W \coloneqq \{ W(f) \}_{f = 1}^F$.

The optimization of Semi-IVE is also based on the MM algorithm described in Section~\ref{sec:BSE:MM}.
The surrogate function for Semi-IVE is the same as that of IVE, i.e., Eq. \eqref{eq:loss:MM},
and the problem of minimizing the surrogate function becomes the following constrained optimization problem for each frequency bin $f = 1,\ldots,F$:
\begin{align}
\label{problem:semi-BSE}
\left.
\begin{array}{cl}
\underset{W(f)}{\minimize} & g (W(f), r) \quad \text{(defined by \eqref{eq:loss:MM})} \\
\subject-to & W(f)^h A_1(f) = E_1.
\end{array}
\right\}
\end{align}
The goal of this section is to develop an efficient algorithm to obtain a (local) optimal solution for problem~\eqref{problem:semi-BSE}.

Hereafter, we omit frequency bin index $f$ and use the following notations for simplicity:
\begin{align}
\label{eq:E1}
E_1 &= [\, \bm{e}_1^{(M)}, \ldots,\bm{e}_L^{(M)} \,] = \begin{bmatrix}
I_L \\
O_{M - L, L}
\end{bmatrix} \in \mathbb{C}^{M \times L},
\\
\label{eq:E2}
E_2 &= [\, \bm{e}_{L + 1}^{(M)}, \ldots,\bm{e}_M^{(M)} \,] = \begin{bmatrix}
O_{L, M - L} \\
I_{M - L}
\end{bmatrix} \in \mathbb{C}^{M \times (M - L)},
\\
\label{eq:W2}
W_2 &= [\, \bm{w}_{L+1},\ldots,\bm{w}_K, W_{\z} \,] \in \mathbb{C}^{M \times (M - L)}.
\end{align}
Here, the superscript $\empty^{(d)}$ in $\bm{e}_i^{(d)} \in \mathbb{C}^d$ represents the length of the unit vectors.

\subsection{LCMV beamforming algorithm for $\bm{w}_1,\ldots,\bm{w}_L$}
\label{sec:semi-BSE:LCMV}

By applying Proposition~\ref{prop:det} in Appendix~\ref{appendix:lemma}, the objective function can be expressed as
\begin{alignat}{3}
\nonumber
g(W) &=&\,& \sum_{i = 1}^K \bm{w}_i^h V_i \bm{w}_i + \trace (W_{\z}^h V_{\z} W_{\z})
\\
\label{eq:loss:semi-BSE}
&&&+ \log \det \left( A_1^h A_1 \right) - \log \det \left( W_2^h W_2 \right).
\end{alignat}
Hence, in problem \eqref{problem:semi-BSE}, separation filters $\bm{w}_1,\ldots,\bm{w}_L$ that correspond to the given transfer functions $\bm{a}_1,\ldots,\bm{a}_L$ can be globally optimized by solving 
\begin{align}
\label{problem:BF:W}
\tag{LCMV}
\left.
\begin{array}{cl}
\underset{\bm{w}_i}{\minimize} & \bm{w}_i^h V_i \bm{w}_i \\
\subject-to & \bm{w}_i^h A_1 = (\bm{e}_i^{(L)})^\top \in \mathbb{C}^{1 \times L}.
\end{array}
\right\}
\end{align}
This problem is nothing but the LCMV beamforming problem that is solved as
\begin{align}
\bm{w}_i = V_i^{-1} A_1 \left(A_1^h V_i^{-1} A_1  \right)^{-1} \bm{e}_i^{(L)} \in \mathbb{C}^M.
\end{align}

\subsection{Block coordinate descent (BCD) algorithm for $W_2$}
\label{sec:Semi-IVE}

We develop a BCD algorithm for optimizing the remaining variables $W_2$.
Since $W_2 \in \mathbb{C}^{M \times (M - L)}$ is restricted by the $(M - L)L$ linear constraints $W_2^h A_1 = O_{M - L, L}$,
it can be parameterized using the $(M - L)^2$ variables.
One such choice is given by
\begin{align}
\label{eq:semi-BSE:W2}
W_2 &= W_2' \overline{W} \in \mathbb{C}^{M \times (M - L)}, \quad \overline{W} \in \mathbb{C}^{(M - L) \times (M - L)},
\\
\label{eq:semi-BSE:W2'}
W_2' &= [A_1, E_2]^{-h} E_2 \in \mathbb{C}^{M \times (M - L)},
\end{align}
where $E_2$ is defined as \eqref{eq:E2}, and it is assumed that $[A_1, E_2]$ is nonsingular.
It is easy to see that this $W_2$ certainly satisfies the linear constraints.
By substituting \eqref{eq:semi-BSE:W2}--\eqref{eq:semi-BSE:W2'} into \eqref{eq:loss:semi-BSE},
we can reformulate problem \eqref{problem:semi-BSE} as an unconstrained optimization problem of minimizing
\begin{align*}
\overline{g} (\overline{W}, r)
= \sum_{i = L + 1}^K \overline{\bm{w}}_i^h \overline{V}_{\!i} \overline{\bm{w}}_i
+ \trace (\overline{W}_{\!\z}^h \overline{V}_{\!\z} \overline{W}_{\!\z})
- 2 \log |\det \overline{W} |
\end{align*}
with respect to $\overline{W}$, where we define
\begin{align}
\overline{W} &= [\, \overline{\bm{w}}_{L + 1}, \ldots, \overline{\bm{w}}_K, \overline{W}_{\!\z} \,] \in \mathbb{C}^{(M - L) \times (M - L)},
\\
\overline{\bm{w}}_i &\in \mathbb{C}^{(M - L) \times 1}, \quad i = L + 1, \ldots, K,
\\
\overline{W}_{\!\z} &\in \mathbb{C}^{(M - L) \times (M - K)},
\\
\overline{V}_{\!i} &= \left( W_2' \right)^h V_i W_2' \in \mathcal{S}_{++}^{M - L}, ~~~ i \in \{ L + 1, \ldots, K, \z \}.
\end{align}
Interestingly, this problem is nothing but the BSE problem and was already discussed in Section~\ref{sec:BSE}.

In a similar manner to IVE-IP2, our proposed Semi-IVE algorithm updates $\overline{\bm{w}}_{L + 1},\ldots,\overline{\bm{w}}_{K}$ one by one by solving the following subproblem
for each $i = L + 1, \ldots, K$:
\begin{align}
\label{problem:Semi-IVE}
(\overline{\bm{w}_i}, \overline{W}_{\!\z}) \in \argmin_{\overline{\bm{w}}_i,\,\overline{W}_{\!\z}} \overline{g} (\overline{\bm{w}}_{L + 1},\ldots,\overline{\bm{w}}_K,\overline{W}_{\!\z}, r).
\end{align}
\begin{itemize}
\item When $L = K - 1$ and $\overline{W} = [\, \overline{\bm{w}}_K, \overline{W}_{\!\z} \,]$,
a global optimal solution for problem \eqref{problem:Semi-IVE}
can be obtained by applying Proposition~\ref{prop:IVE-IP2:K=1}.
\item When $1 \leq L \leq K - 2$,
a global optimal solution of \eqref{problem:Semi-IVE} can be obtained by applying Proposition~\ref{prop:IVE-IP2:K>1}.

\item Note that for the same reason as in the proposed IVE-IP2,
updating $\overline{W}_{\!\z}$ in \eqref{problem:Semi-IVE} is not necessary.
\end{itemize}
In summary, Semi-IVE can be presented in Algorithm~\ref{alg:Semi-IVE}, in which we use the following notations:
\begin{align}
I_{M - L} &= [\, \overline{\bm{e}}_{L + 1}, \ldots, \overline{\bm{e}}_K, \overline{E}_{\z} \,] \in \GL (M - L),
\\
\overline{\bm{e}}_i &= \bm{e}_{i - L}^{(M - L)} \in \mathbb{C}^{M - L}, \quad i = L + 1, \ldots, K,
\\
\overline{E}_{\z} &= [\, \bm{e}_{K - L + 1}^{(M - L)}, \ldots, \bm{e}_{M - L}^{(M - L)} \,] \in \mathbb{C}^{(M - L) \times (M - K)}.
\end{align}
\begin{remark}
Semi-IVE does not update $\overline{W}_{\!\z}$ and $W_{\z}$ from the initial values.
Hence, for the same reason as in Remark~\ref{remark:projection-back}, we need to optimize $\image \overline{W}_{\!\z}$ and $\image W_{\z}$ after the optimization and before performing the projection back operation.
For this purpose, we adopt the following formula:
\begin{align}
\label{eq:Semi-IVE:Im(Wz)}
W_{\z} = W_2' \overline{W}_{\!\z},\quad
\overline{W}_{\!\z} = \begin{bmatrix}
(\overline{W}_{\!\s}^h \overline{V}_{\!\z} \overline{E}_{\s})^{-1} (\overline{W}_{\!\s}^h \overline{V}_{\!\z} \overline{E}_{\z}) \\
-I_{M - K}
\end{bmatrix}.
\end{align}
Here, we use the following notations:
\begin{alignat*}{3}
\overline{W} &= [\, \overline{W}_{\!\s}, \overline{W}_{\!\z} \,],
&\quad
\overline{W}_{\!\s} &= [\, \overline{\bm{w}}_{L + 1}, \ldots, \overline{\bm{w}}_{K} \,],
\\
I_{M - L} &= [\, \overline{E}_{\s}, \overline{E}_{\z} \,],
&\quad
\overline{E}_{\s} &= [\, \overline{\bm{e}}_{L + 1}, \ldots, \overline{\bm{e}}_{K} \,].
\end{alignat*}
The formula \eqref{eq:Semi-IVE:Im(Wz)} optimizes $\image \overline{W}_{\!\z}$ (and hence $\image W_{\z}$)
as described in Remark~\ref{remark:projection-back}.
\end{remark}

\section{Computational time complexity}
\label{sec:computational-complexity}

We compare the computational time complexity per iteration of the algorithms presented in Table~\ref{table:alg}.
In IVE-IP1, IVE-IP2, and Semi-IVE, the runtime is dominated by
\begin{itemize}
\item the computation of $V_1,\ldots,V_K \in \mathcal{S}_{++}^M$ at Step \ref{step:V} in Algorithm~\ref{alg:main}, which costs $\mathrm{O}(K M^2 FT)$; or
\item the matrix inversions and generalized eigenvalue decompositions in Algorithms~\ref{alg:IVE-IP1}--\ref{alg:Semi-IVE}, which cost
$\mathrm{O}(K M^3 F)$.
\end{itemize}
Thus, the time complexity of these algorithms is 
\begin{align}
\mathrm{O}(K M^2 FT + K M^3 F).
\end{align}
On the other hand, the time complexity of IVA-IP1 and IVA-IP2 is
\begin{align}
\mathrm{O}(M^3 FT + M^4 F)
\end{align}
due to the computation of $M$ covariance matrices $V_1,\ldots,V_M$.
Consequently, IVE reduces the computational time complexity of IVA by a factor of $K / M$.

\begin{algorithm}[p]
{
\caption{IVE ($L = 0$) and Semi-IVE ($1 \leq L \leq K$)
based on generalized Gaussian distributions in Assumption~\ref{assumption:superGaussian}:
$G(\| \bm{s}_{i}(t) \|, \alpha_i) = \left( \frac{\| \bm{s}_i(t) \|}{\alpha_i} \right)^{\beta} + 2 F \log \alpha_i$.
Here, $\alpha_1,\ldots,\alpha_K$ are parameters to be optimized.
}
\label{alg:main}
\DontPrintSemicolon
{\setstretch{1.3}
\nl\KwData{Observed signal $\bm{x}$;}
\myinput{Number of sources $K$; and}
\myinput{$L$ transfer function $A_1 = [\bm{a}_1,\ldots,\bm{a}_L]$,}
\myinput{where $0 \leq L \leq K$.}
\nl\KwResult{Separated spatial images $\bm{x}_1,\ldots,\bm{x}_K$.}
}
{\setstretch{1.4}
\nl\Begin{
\tcc{Initialization}
\nl $W(f) \leftarrow -I_M$\;
\nl $V_{\z}(f) \leftarrow \frac{1}{T} \sum_{t = 1}^T \bm{x}(f,t) \bm{x}(f,t)^h$\;
\nl \If{using IVE-IP1 or IVE-IP2}{
    \nl Update $W_{\z}(f)$ using Step \ref{step:IVE-IP1:Wz} in Algorithm~\ref{alg:IVE-IP1}
}
\nl \If{using Semi-IVE}{
    \nl $W_2'(f) \leftarrow [\, A_1(f), E_2 \,]^{-h} E_2$ \;
    \nl $\overline{V}_{\!\z}(f) \leftarrow W_2'(f)^h V_{\z}(f) W_2'(f)$ \;
    \nl Update $W_{\z}(f)$ using \eqref{eq:Semi-IVE:Im(Wz)}\;
}
\tcc{Optimization}
\nl \Repeat{convergence}{
    \nl \For{$i = 1,\ldots,K$}{
    	\nl $s_i(f,t) \leftarrow \bm{w}_i(f)^h \bm{x}(f,t)$\;
    	\nl $r_i(t) \leftarrow \| \bm{s}_i(t) \|$\;
    	\nl $\alpha_i^{\beta} = \frac{\beta}{2F} (\frac{1}{T} \sum_t r_i(t)^\beta)$\;
    	\label{step:alpha}
    	\nl $\phi_i(t) \leftarrow \frac{G'(r_i(t), \alpha_i)}{2 r_i(t)}
    	    = \frac{\beta}{2} \frac{1}{\alpha_i^\beta r_i(t)^{2 - \beta}}$\;
    	\nl $\phi_i(t) \leftarrow \min \{ \phi_i(t), 10^5 \times \min \{ \phi_i(t) \}_{t = 1}^T \}$
    	\tcp{for numerical stability}
        \label{step:G'}
    	\nl $V_i(f) \leftarrow \frac{1}{T} \sum_t \phi_i(t) \bm{x}(f,t) \bm{x}(f,t)^h$\;
    	\label{step:V}
    	\nl $V_i(f) \leftarrow V_i(f) + 10^{-3} \trace\{ V_i(f) \} I_M$
    	\tcp{for numerical stability}
    }
    \nl Update $W(f)$ for each $f$ using IVE-IP1, IVE-IP2, or Semi-IVE.\;
}
\tcc{Source extraction}
\nl \If{using IVE-IP2}{
    \nl Update $W_{\z}(f)$ using Step \ref{step:IVE-IP1:Wz} in Algorithm~\ref{alg:IVE-IP1}
    \label{step:Im(Wz)}
}
\nl \If{using Semi-IVE}{
    \nl Update $W_{\z}(f)$ using \eqref{eq:Semi-IVE:Im(Wz)}\;
}
\nl $\bm{x}_i(f,t) \leftarrow \left( W(f)^{-h} \bm{e}_i \right) \bm{w}_i(f)^h  \bm{x}(f,t)$\;
}
}
}
\end{algorithm}
\begin{algorithm}[p]
\caption{IVE-IP1 (proposed in~\cite{scheibler2019overiva})}
\label{alg:IVE-IP1}
\setstretch{1.2}
\DontPrintSemicolon
\nl \For{$i = 1,\ldots,K$}{
\nl $\bm{u}_i(f) \leftarrow \left( W(f)^h V_i(f) \right)^{-1} \bm{e}_i$\;
\nl $\bm{w}_i(f) \leftarrow \bm{u}_i(f) \left( \bm{u}_i(f)^h V_i(f) \bm{u}_i(f) \right)^{-\frac{1}{2}}$\;
\nl
\label{step:IVE-IP1:Wz}
$W_{\z}(f) \leftarrow 
{\small
\begin{bmatrix}
(W_{\s}(f)^h V_{\z}(f) E_{\s})^{-1} (W_{\s}(f)^h V_{\z}(f) E_{\z}) \\
-I_{M - K}
\end{bmatrix}
}$\;
}
\end{algorithm}
\begin{algorithm}[p]
{
\caption{IVE-IP2 for $K = 1$}
\label{alg:IVE-IP2:K=1}
\setstretch{1.2}
\DontPrintSemicolon
\nl Solve $V_{\z}(f) \bm{u} = \lambda_{\max} V_1(f) \bm{u}$ to obtain the eigenvector $\bm{u}$ corresponding to the largest eigenvalue $\lambda_{\max} $. \;
\nl $\bm{w}_1(f) \leftarrow \bm{u} \left( \bm{u}^h V_1(f) \bm{u}  \right)^{- \frac{1}{2}}$\;
}
\end{algorithm}
\begin{algorithm}[p]
{
\caption{IVE-IP2 for $K \geq 2$}
\label{alg:IVE-IP2:K>1}
\setstretch{1.2}
\DontPrintSemicolon
\nl \For{$i = 1,\ldots,K$}{
\nl \For{$\ell \in \{i, \z \}$}{
	\nl $P_\ell(f) \leftarrow \left( W(f)^h V_\ell(f) \right)^{-1} [\, \bm{e}_i, E_{\z} \,]$ \;
	\nl $G_\ell(f) \leftarrow P_\ell(f)^h V_\ell(f) P_\ell(f)$ \;
}
\nl Solve $G_i(f) \bm{b} = \lambda_{\max} G_{\z}(f) \bm{b}$ to obtain $\bm{b}$ corresponding to the largest eigenvalue $\lambda_{\max}$. \;
\nl $\bm{w}_i(f) \leftarrow P_i(f) \bm{b} \left( \bm{b}^h G_i(f) \bm{b} \right)^{-\frac{1}{2}}$
}
}
\end{algorithm}
\begin{algorithm}[p]
{
\caption{Semi-IVE}
\label{alg:Semi-IVE}
\setstretch{1.2}
\DontPrintSemicolon
\tcc{LCMV beamforming}
\nl \For{$i = 1,\ldots,L$}{
	\nl {\small $\bm{w}_i(f) \leftarrow V_i(f)^{-1} A_1(f) \left( A_1(f)^h V_i(f)^{-1} A_1(f) \right)^{-1} \bm{e}_i$}\;
}
\nl \If{$L = K$}{
	return
}

\tcc{BCD}
\nl \For{$i = L + 1,\ldots,K$}{
\nl $\overline{V}_{\!i}(f) \leftarrow W_2'(f)^h V_i(f) W_2'(f)$ \;
}
\nl \If{$L = K - 1$}{
	\nl Solve $\overline{V}_{\!\z}(f) \overline{\bm{u}} = \lambda_{\max} \overline{V}_{\!K}(f) \overline{\bm{u}}$ to obtain $\overline{\bm{u}}$
		corresponding to the largest eigenvalue $\lambda_{\max}$. \;
	\nl $\bm{w}_K(f) \leftarrow W_2'(f) \overline{\bm{u}} \left( \overline{\bm{u}}^h \overline{V}_{\!K}(f) \overline{\bm{u}} \right)^{-\frac{1}{2}}$
}
\nl \Else{
	\nl \For{$i = L + 1,\ldots,K$}{
		\nl \For{$\ell \in \{i, \z \}$}{
			\nl $\overline{P}_\ell(f) \leftarrow \left( \overline{W}(f)^h \overline{V}_\ell(f) \right)^{-1} [\, \overline{\bm{e}}_i, \overline{E}_{\z} \,]$ \;
			\nl $\overline{G}_\ell(f) \leftarrow \overline{P}_\ell(f)^h \overline{V}_\ell(f) \overline{P}_\ell(f)$ \;
		}
		\nl Solve $\overline{G}_i(f) \overline{\bm{b}} = \lambda_{\max} \overline{G}_{\z}(f) \overline{\bm{b}}$ to obtain $\overline{\bm{b}}$ corresponding to the largest eigenvalue $\lambda_{\max}$. \;
		\nl $\bm{w}_i(f) \leftarrow W_2'(f) \overline{P}_i(f) \overline{\bm{b}} \left( \overline{\bm{b}}^h \overline{G}_i(f) \overline{\bm{b}} \right)^{-\frac{1}{2}}$
}
}
}
\end{algorithm}

\section{Experiments}
\label{sec:exp}

In this numerical experiment, we evaluated the following properties of the IVE and Semi-IVE algorithms: 
\begin{itemize}
\item The source extraction performance for the speech signals in terms of the signal-to-distortion ratio (SDR~\cite{vincent2006sdr}) between the estimated and oracle source spatial images;
\item the runtime performance.
\end{itemize}
We compared the performance of the following five methods whose optimization procedures are summarized in Table~\ref{table:alg}:
\begin{enumerate}
\item \textbf{IVA-IP1-old}:
The conventional AuxIVA with IP1~\cite{ono2011auxiva}, followed by picking $K$ signals in an oracle manner.
\item \textbf{IVE-IP1-old}:
The conventional IVE-IP1 or IVE-OC~\cite{scheibler2019overiva} (Algorithms~\ref{alg:main} and \ref{alg:IVE-IP1}).
\item \textbf{IVE-IP2-old}:
The conventional IVE-IP2~\cite{scheibler2020ive}
(see~{\cite[Algorithm 2]{scheibler2020ive}} for the implementation when $K \geq 2$).
\item \textbf{IVE-IP2-new}:
The proposed IVE-IP2 (Algorithms~\ref{alg:main}, \ref{alg:IVE-IP2:K=1}, and \ref{alg:IVE-IP2:K>1}),
where the generalized eigenvalue problems will be solved using the power method with 30 iterations
(see Section~\ref{sec:exp-implementation} for details).%
\footnote{
Note that in the case of $K = 1$, IVE-IP2-old (called FIVE~\cite{scheibler2020fast}) and IVE-IP2-new are the same algorithm (see Section~\ref{sec:IVE-IP2:K=1}).
However, we distinguish between them because we propose to use the power method to efficiently obtain the first generalized eigenvector in the proposed IVE-IP2-new. 
}
\item \textbf{Semi-IVE-$(L)$-new}:
The proposed semiblind IVE algorithm, where the transfer functions of $L$ super-Gaussian sources are given as an oracle
(Algorithms~\ref{alg:main} and \ref{alg:Semi-IVE}).
\end{enumerate}

\subsection{Dataset}
\label{sec:exp-data}

As evaluation data, we generated synthesized convolutive noisy mixtures of speech signals.

\textit{Room impulse response (RIR) data}:
We used the RIR data recorded in room \textsf{E2A} from the RWCP Sound Scene Database in Real Acoustical Environments~\cite{rwcp}.
The reverberation time ($\mathrm{RT}_{60}$) of room \textsf{E2A} is 300 ms.
These data consist of nine RIRs from nine different directions.

\textit{Speech data}:
We used point-source speech signals from the test set of the TIMIT corpus~\cite{timit}.
We concatenated the speech signals from the same speaker so that the length of each signal exceeded ten seconds. 
We prepared 168 speech signals in total.

\textit{Noise data}:
We used a background noise signal recorded in a cafe (\textsf{CAF}) from the third `CHiME' Speech Separation and Recognition Challenge~\cite{chime3}.
We chose a monaural signal captured at ch1 on a tablet and used it as a point-source noise signal.
The noise signal was about 35 minutes long.

\textit{Mixture signals}:
We generated 100 mixtures consisting of $K \in \{1, 2, 3\}$ speech signals and $J = 5$ noise signals:
\begin{enumerate}
\item We selected $K$ speech signals at random from the 168 speech data.
We selected $J$ non-overlapping segments at random from the noise data and prepared $J$ noise signals.
We selected $K + J$ RIRs at random from the original nine RIRs.
\item We convolved the $K$ speech and $J$ noise signals with the selected $K + J$ RIRs to create $K + J$ spatial images. 
\item We added the obtained $K + J$ spatial images in such a way that 
$\mathrm{SNR} \coloneqq 10 \log_{10} \frac{ \frac{1}{K} \sum_{i = 1}^K \lambda_i^{(\mathrm{s})} }{ \sum_{j = 1}^J \lambda_j^{(\mathrm{n})} }$ becomes 0 dB if $K = 1,2$ and 5 dB if $K = 3$,
where $\lambda_i^{(\mathrm{s})}$ and $\lambda_j^{(\mathrm{n})}$ denote the sample variances of the $i$th speech-source and $j$th noise-source spatial images.
\end{enumerate}

\subsection{Experimental conditions}
\label{sec:exp-cond}

For all the methods, we initialized $W(f) = - I_M$ and
assumed that each super-Gaussian source $\bm{s}_i(t) \in \mathbb{C}^F$ follows the generalized Gaussian distribution (GGD).
More concretely, we set to $G(r_i(t), \alpha_i) = (\frac{r_i(t)}{\alpha_i})^{\beta} + 2 F \log \alpha_i$ with shape parameter $\beta = 0.1$ in Assumption~\ref{assumption:superGaussian}.
Scale parameters $\alpha_1,\ldots,\alpha_K \in \mathbb{R}_{> 0}$ are optimized to
$\alpha_i^{\beta} = \frac{\beta}{2F} (\frac{1}{T} \sum_t r_i(t)^\beta)$
every after $W$ is updated (Step~\ref{step:alpha} in Algorithm~\ref{alg:main}).

In Semi-IVE-$(L)$-new,
we prepared oracle transfer functions $\bm{a}_1,\ldots,\bm{a}_L$ using oracle spatial images $\bm{x}_1,\ldots,\bm{x}_L$. 
We set $\bm{a}_\ell (f)$ to the first eigenvector of 
sample covariance matrix 
$R_\ell(f) = \frac{1}{T} \sum_{t = 1}^T \bm{x}_\ell(f,t) \bm{x}_\ell(f,t)^h$
for each source $\ell \in \{ 1,\ldots,L \}$.
Here, $\bm{a}_\ell(f)$ was normalized to be $\| \bm{a}_\ell(f) \| = 1$.

The performance was tested for one to three speakers and two to eight microphones.
The sampling rate was 16 kHz,
the reverberation time was 300 ms,
the frame length was 4096 (256 ms),
and the frame shift was 1024 (64 ms).

\subsection{Implementation notes}
\label{sec:exp-implementation}

We implemented all the algorithms in Python 3.7.1.

In IVE-IP2-new and Semi-IVE, we need to solve the generalized eigenvalue problems of form
$A \bm{x} = \lambda_{\mathrm{max}} B \bm{x}$ where we require only the first generalized eigenvector.
To prevent `for loops' in the Python implementation, we solved them 
(i) by first transforming the problem into an eigenvalue problem
$\left( B^{-1} A \right) \bm{x} = \lambda_{\mathrm{max}} \bm{x}$,
and then (ii) using power iteration (also known as the power method~\cite{atkinson2008}) to obtain the required first eigenvector.
The number of iterations in the power method was set to 30 in this experiment.

On the other hand, in IVE-IP2-old for $K \geq 2$, all the generalized eigenvectors have to be obtained.
To prevent `for loops,' we implemented it
(i) by first transforming the generalized eigenvector problem into an eigenvector problem in the same way as above,
and then (ii) by calling the \textsf{numpy.linalg.eig} function.
IVE-IP2-old for $K = 1$ (i.e., FIVE~\cite{scheibler2020fast}) was implemented in the same way.

These implementations are not recommended for numerical stability,
but we adopted them to make the algorithms run fast.

\begin{figure*}[p]
\begin{center}
\begin{subfigure}{1.0\textwidth}
\centering
\includegraphics[width=\textwidth]{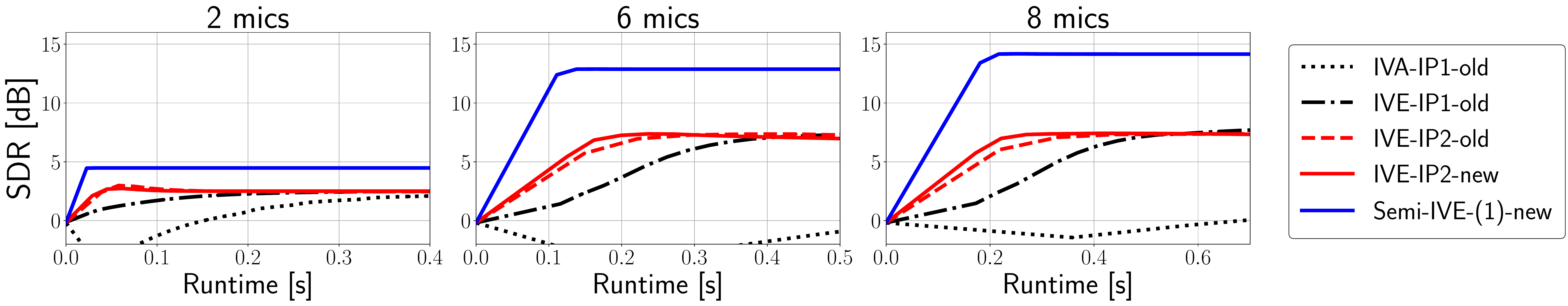}
\caption{One target source ($K = 1$), SNR = 0 [dB], averaged signal length: 11.46 sec}
\vspace{3 mm}
\end{subfigure}
\begin{subfigure}{1.0\textwidth}
\centering
\includegraphics[width=\textwidth]{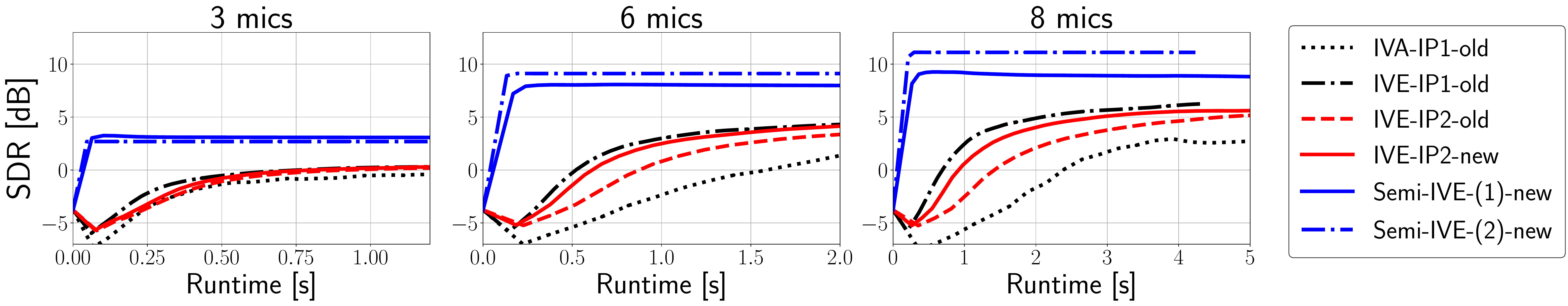}
\caption{Two target sources ($K = 2$), SNR = 0 [dB], averaged signal length: 12.85 sec}
\vspace{3 mm}
\end{subfigure}
\begin{subfigure}{1.0\textwidth}
\centering
\includegraphics[width=\textwidth]{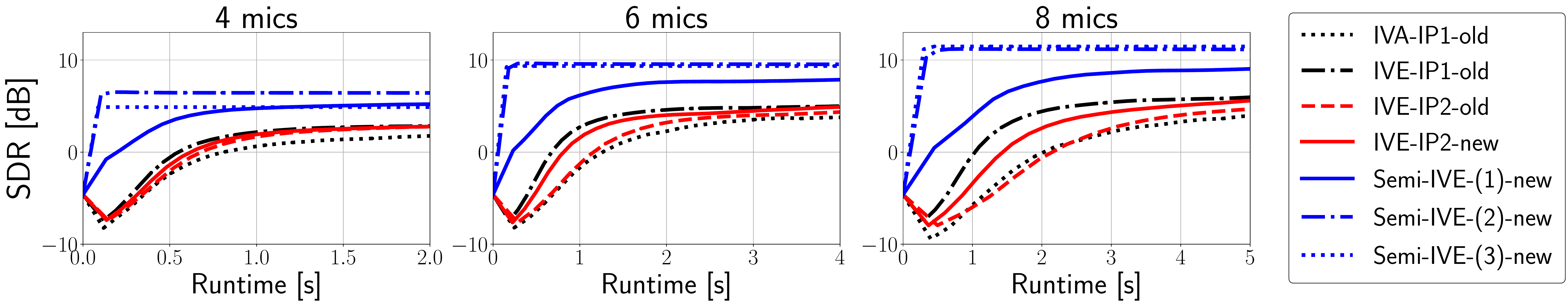}
\caption{Three target sources ($K = 3$), SNR = 5 [dB], averaged signal length: 12.92 sec}
\vspace{3 mm}
\end{subfigure}
\end{center}
\vspace{-4 mm}
\caption{
SDR [dB] performance as a function of runtime.
Results shown here are averaged over 100 mixtures.
}
\label{fig:SDR}
\end{figure*}
\begin{table*}[p]
\begin{center}
{
\caption{Realtime factors of algorithms when number of iterations is set to 50}
\label{table:RTF}
\begin{tabular}{c|ccc|ccc|ccc} \hline
Number of sources & \multicolumn{3}{c|}{$K = 1$} & \multicolumn{3}{c|}{$K = 2$} & \multicolumn{3}{c}{$K = 3$}
\\ \hline
Number of microphones & $M = 2$ & $M = 6$ & $M = 8$ & $M = 3$ & $M = 6$ & $M = 8$ & $M = 4$ & $M = 6$ & $M = 8$
\\ \hline
Signal length [sec] & 11.77 & 11.77 & 11.77 & 12.43 & 12.43 & 12.43 & 12.81 & 12.81 & 12.81
\\ \hline \hline
IVA-IP1-old & 0.08 & 0.58 & 0.99 & 0.18 & 0.61 & 1.05 & 0.31 & 0.70 & 1.20
\\
IVE-IP1-old & 0.05 & 0.13 & 0.18 & 0.14 & 0.27 & 0.36 & 0.28 & 0.44 & 0.67
\\
IVE-IP2-old & 0.07 & 0.32 & 0.48 & 0.21 & 0.59 & 0.93 & 0.38 & 0.81 & 1.45
\\
IVE-IP2-new & 0.06 & 0.16 & 0.22 & 0.19 & 0.40 & 0.59 & 0.35 & 0.59 & 1.00
\\ \hline
Semi-IVE-$(1)$-new & 0.04 & 0.12 & 0.17 & 0.15 & 0.29 & 0.41 & 0.32 & 0.54 & 0.89
\\
Semi-IVE-$(2)$-new & - & - & - & 0.14 & 0.26 & 0.35 & 0.28 & 0.44 & 0.70
\\
Semi-IVE-$(3)$-new & - & - & - & - & - & - & 0.28 & 0.42 & 0.65
\\ \hline
\end{tabular}
}
\end{center}
\end{table*}

\subsection{Experimental results}
\label{sec:exp-res}

Figure~\ref{fig:SDR} shows the SDR performance of each algorithm as a function of the runtime,
and Table~\ref{table:RTF} presents the realtime factor (RTF) of the algorithms when the number of iterations is set to 50 for all the algorithms.
These results were obtained by running the algorithms on a PC with 
Intel(R) Core(TM) i7-7820X CPU @ 3.60GHz using a single thread.

\subsubsection{Effectiveness of IVE-IP2-new for extracting unique target source ($K = 1$)}
In Fig.~\ref{fig:SDR}(a), IVE-IP2-new showed faster convergence than the conventional algorithms while achieving a similar SDR performance.
Interestingly, the full-blind IVE algorithms in the $M \geq 6$ cases produced rather better SDRs than the informed Semi-IVE-$(1)$-new in the $M = 2$ case.
This result motivates us to use more microphones at the expense of increased runtime.
As desired, these increased runtime of the IVE algorithms, especially the proposed IVE-IP2-new, was shown to be very small compared to IVA-IP1-old.

\subsubsection{Effectiveness of IVE-IP2-new compared to IVE-IP2-old}
From Fig.~\ref{fig:SDR} and Table~\ref{table:RTF}, it can be confirmed for $M \geq 6$ that the computational cost of the proposed IVE-IP2-new is consistently smaller than IVE-IP2-old while maintaining its separation performance,
which shows the effectiveness of IVE-IP2-new.
On the other hand, if $M$ is small, IVE-IP2-new gave almost the same performance as IVE-IP2-old.

\subsubsection{Comparison of IVA-IP1-old, IVE-IP1-old, and IVE-IP2-new}
If $K = 1$ (Fig.~\ref{fig:SDR}(a)), the proposed IVE-IP2-new gave the fastest convergence among other blind algorithms, which clearly shows the superiority of IVE-IP2-new.
On the other hand, if $K \in \{2, 3\}$ (Fig.~\ref{fig:SDR}(b)--(c)), the convergence of IVE-IP2-new is slower than that of IVE-IP1-old,
mainly because the computational cost per iteration of IVE-IP2-new exceedes that of IVE-IP1-old (see Table~\ref{table:RTF}).
Therefore, IVE-IP1-old is still important when extracting more than two sources.

\subsubsection{Effectiveness of Semi-IVE-$(L)$-new}
In Fig.~\ref{fig:SDR}, the proposed Semi-IVE algorithms naturally outperformed all of the full-blind IVA and IVE algorithms.
Surprisingly, when the $L \coloneqq K - 1$ transfer functions are given a priori (and $M \geq 4$),
the Semi-IVE algorithms (Semi-IVE-$(1)$-new if $K = 2$ and Semi-IVE-$(2)$-new if $K = 3$)
achieved comparable or sometimes better performance than Semi-IVE-$(K)$-new.
The convergence of Semi-IVE-$(L)$-new with $L \geq K - 1$ was also extremely fast.
These results clearly show the effectiveness of the proposed Semi-IVE algorithms.

\section{Concluding remarks}
\label{sec:conclusion}

We presented two new efficient BCD algorithms for IVE:
(i) IVE-IP2, which extracts all the super-Gaussian sources from a linear mixture in a fully blind manner, and
(ii) Semi-IVE, which improves IVE-IP2 in the semiblind scenario in which the transfer functions for several super-Gaussian sources are available as prior knowledge.
We also argued that the conventional IVE that relies on the orthogonal constraint (IVE-OC) can be interpreted as BCD for IVE (IVE-IP1).
Due to the stationary Gaussian noise assumption,
these BCD (or IP) algorithms can skip most of the optimization of the filters for separating the noise components, which plays a central role for achieving a low computational cost of optimization.

Our numerical experiment, which extracts speech signals from their noisy reverberant mixture, showed that when $K = 1$ or given at least $K - 1$ transfer functions in a semiblind case, the proposed IVE-IP2 and Semi-IVE resulted in significantly faster convergence compared to the conventional algorithms,
where $K$ is the number of the target-source signals.
The new IVE-IP2 consistently speeded up the old IVE-IP2,
and the conventional IVE-IP1 remains important for extracting multiple sources.

\appendices
\section{}
\label{appendix:lemma}

We prepare several propositions that are needed to rigorously develop the algorithms in Sections~\ref{sec:BSE} and \ref{sec:semi-BSE}.

Proposition~\ref{prop:MM} below gives an inequality which is the basis of the MM algorithm for AuxICA~\cite{ono2010auxica}, AuxIVA~\cite{ono2011auxiva}, and (the auxiliary-function-based) IVE. 
\begin{proposition}[See {\cite[Theorem~1]{ono2010auxica}}]
\label{prop:MM}
Let $G \colon \mathbb{R}_{\geq 0} \to \mathbb{R}$ be differentiable and satisfy
that $\frac{G'(r)}{r}$ is nonincreasing on $r \in \mathbb{R}_{> 0}$.
Then, for arbitrary $r, \widetilde{r} \in \mathbb{R}_{> 0}$, it holds that
\begin{align}
G(r) \leq \frac{G'(\widetilde{r})}{2\widetilde{r}} \cdot r^2
+ \left( G(\widetilde{r}) - \frac{\widetilde{r} \cdot G'(\widetilde{r})}{2} \right).
\end{align}
The inequality holds with equality if $r = \widetilde{r}$.
\end{proposition}

To give a proof of Propositions~\ref{prop:IVE-IP2:K=1}, \ref{prop:IVE-IP2:K>1}, and \ref{prop:IVE-IP1:Wz}, we need the following proposition that provides a sufficient condition for the existence of globally optimal solutions.
\begin{proposition}
\label{prop:loss:lower-bounded}
Suppose $V_1,\ldots,V_n \in \mathcal{S}_{++}^M$ and $1 \leq n \leq M$.
Let $W_2 \in \mathbb{C}^{M \times (M - n)}$ be full column rank.
The function $g$ with respect to
$W_1 = [\bm{w}_{1},\ldots,\bm{w}_n] \in \mathbb{C}^{M \times n}$,
defined by
\begin{align*}
\nonumber
g(W_1) &= 
\sum_{i = 1}^n \bm{w}_i^h V_i \bm{w}_i - \log | \det [W_1,W_2] |^2,
\end{align*}
 is lower bounded and attains its minimum at some $W_1$.
\end{proposition}
\begin{proof}
The statement for $n = M$ was proved in~\cite{degerine2004,degerine2006maxdet,yeredor2009HEAD,yeredor2012SeDJoCo}.
Since that for $1 \leq n < M$ can also be proved in the same way as the former part, we omit the proof.
\end{proof}
Proposition~\ref{prop:det} below gives a useful expression of $|\det W|^2$.
The same statement for $L = 1$ in Proposition~\ref{prop:det} can be found in \cite{nakatani2019wpd-ml},
and we here extend it for general $1 \leq L \leq M$.
\begin{proposition}
\label{prop:det}
Let
$1 \leq L < M$ and
\begin{align*}
W &= \begin{bmatrix} W_1, W_2\end{bmatrix} \in \mathbb{C}^{M \times M},
\\
W_1 &\in \mathbb{C}^{M \times L}, \quad W_2 \in \mathbb{C}^{M \times (M - L)},
\\
A_1 &\in \mathbb{C}^{M \times L}, \quad W_1^h A_1 = I_L, \quad W_2^h A_1 = O_{M - L,L}.
\end{align*}
Then, $|\det W |^2 = \det \left( A_1^h A_1 \right)^{-1} \det \left( W_2^h W_2 \right)$ holds.
\end{proposition}
\begin{proof}
The orthogonal projection onto $\image A_1$ is given by
$P = A_1 (A_1^h A_1)^{-1} A_1^h \in \mathbb{C}^{M \times M}$.
Then, it holds that
\begin{align*}
| \det W |^2 
&= \left| \det ([\, P W_1 + (I - P) W_1, W_2 \,]) \right|^2
\\
&= \left| \det ([\, P W_1, W_2 \,]) \right|^2
\\
&= \det ([\, P W_1, W_2 \,]^h [\, P W_1, W_2 \,])
\\
&= \det \left(
\begin{bmatrix}
W_1^h P^h P W_1 & O_{L, M - L} \\
O_{M - L, L} & W_2^h W_2
\end{bmatrix}
\right)
\\
&= \det (A_1^h A_1)^{-1} \cdot \det (W_2^h W_2),
\end{align*}
where we use $\image (I - P) \supseteq \image W_2$ in the second equality.
\end{proof}

\section{Proof of Proposition~\ref{prop:IVE-IP2:K>1}}
\label{appendix:proof-of-prop3}

\begin{proof}
The proof is very similar to that of \cite[Theorem~3]{scheibler2020ive}.
By Proposition~\ref{prop:loss:lower-bounded}, problem \eqref{problem:IVE-IP2} has global optimal solutions,
and they satisfy the stationary conditions \eqref{eq:KKT:wk}--\eqref{eq:KKT:Wz}.
In Eqs.~\eqref{eq:KKT:wk}--\eqref{eq:KKT:Wz}, the first $K$ rows except for the $i$th row is linear with respect to $\bm{w}_i$ and $W_{\z}$,
and so they can be solved as
\begin{alignat}{3}
\label{eq:IVE-IP:K>1:101}
\bm{w}_i &= P_i G_i^{-1} \bm{c}_i,
&\quad
\bm{c}_i &\in \mathbb{C}^{(M - K + 1) \times 1},
\\
\label{eq:IVE-IP:K>1:102}
W_{\z} &= P_{\z} G_{\z}^{-1} C_{\z},
&\quad
C_{\z} &\in \mathbb{C}^{(M - K + 1) \times (M - K)},
\end{alignat}
where $P_\ell$ and $G_\ell$ for $\ell \in \{i, \z \}$ are defined by \eqref{eq:IVE-IP2:K>1:P}--\eqref{eq:IVE-IP2:K>1:G},
and $\bm{c}_i$ and $C_{\z}$ are free variables (the roles of $G_i^{-1}$ and $G_{\z}^{-1}$ will be clear below).
Substituting \eqref{eq:IVE-IP:K>1:101}--\eqref{eq:IVE-IP:K>1:102} into the objective function $g$, we have
\begin{align}
\nonumber
g  &= \bm{c}_i^h G_i^{-1} \bm{c}_i + \trace \{ C_i^h G_{\z}^{-1} C_{\z} \}
\\
\label{eq:prop:phi}
&\quad - \log | \det [\, W_0 \mid P_i G_i^{-1} \bm{c}_i \mid P_{\z} G_{\z}^{-1} C_{\z} \,] |^2 + \const,
\end{align}
where $W_0 \coloneqq [\, \bm{w}_1, \ldots,\bm{w}_{i - 1}, \bm{w}_{i + 1}, \ldots,\bm{w}_K \,] \in \mathbb{C}^{M \times (K - 1)}$.

Now, let us consider minimizing $g$ with respect to $\bm{c}_i$ and $C_{\z}$.
In \eqref{eq:prop:phi}, the $\log \det$ term can be simplified as
\begin{alignat*}{3}
&&&2 \log | \det
\begin{bmatrix}
W_0 \mid P_i G_i^{-1} \bm{c}_i \mid P_{\z} G_{\z}^{-1} C_{\z} 
\end{bmatrix}
|
\\
&=&\,&
2 \log | \det
\begin{bmatrix}
W_0^h
\\
P_i^h V_i
\end{bmatrix}^{-1}
\begin{bmatrix}
W_0^h
\\
P_i^h V_i
\end{bmatrix}
\begin{bmatrix}
W_0 \mid P_i G_i^{-1} \bm{c}_i \mid P_{\z} G_{\z}^{-1} C_{\z} 
\end{bmatrix}
|
\\
&=&~&
2 \log | \det \begin{bmatrix}
W_0^h W_0 & * & * \\
O_{M - K + 1, K - 1} & \bm{c}_i & C_{\z}
\end{bmatrix}
|
+ \const
\\
&=&~&
2 \log | \det [\, \bm{c}_i, C_{\z} \,] | + \const
\end{alignat*}
Here, we used $V_i P_i = \left( W' \right)^{-h} [\, \bm{e}_i, E_{\z} \,] = V_{\z} P_{\z}$ twice in the second equality.
Hence, by applying Proposition~\ref{prop:IVE-IP2:K=1}, $g$ attains its minimum when
\begin{align}
\bm{c}_i &= \bm{u}_i \left( \bm{u}_i^h G_i^{-1} \bm{u}_i  \right)^{- \frac{1}{2}},
\quad C_{\z} = U_{\z} \left( U_{\z}^h G_{\z}^{-1} U_{\z} \right)^{- \frac{1}{2}},
\\
\label{eq:IVE-IP2:prop:Uz}
U_{\z} &\in \mathbb{C}^{(M - K + 1) \times (M - K)} \quad \text{with} \quad U_{\z}^h G_{\z}^{-1} \bm{u}_i = \bm{0},
\\
\label{eq:IVE-IP2:prop:ui}
\bm{u}_i &\in \mathbb{C}^{(M - K + 1) \times 1} \quad \text{with} \quad G_{\z}^{-1} \bm{u}_i = \lambda_{\max} G_i^{-1} \bm{u}_i, 
\end{align}
where $\lambda_{\max}$ in \eqref{eq:IVE-IP2:prop:ui} denotes the largest generalized eigenvalue.
Because \eqref{eq:IVE-IP2:prop:Uz}--\eqref{eq:IVE-IP2:prop:ui}
are equivalent to \eqref{eq:IVE-IP2:K>1:orth}--\eqref{eq:IVE-IP2:K>1:eig} through
$\bm{b}_i = G_i^{-1} \bm{u}_i$ and $B_{\z} = G_{\z}^{-1} U_{\z}$,
a global optimal solution for \eqref{problem:IVE-IP2} can be obtained by
\begin{align*}
\bm{w}_i &= P_i G_i^{-1} \bm{u}_i \left( \bm{u}_i^h G_i^{-1} \bm{u}_i  \right)^{- \frac{1}{2}}
= P_i \bm{b}_i (\bm{b}_i^h G_i \bm{b}_i)^{-\frac{1}{2}},
\\
W_{\z} &= P_{\z} G_{\z}^{-1} U_{\z} \left( U_{\z}^h G_{\z}^{-1} U_{\z} \right)^{- \frac{1}{2}}
= P_{\z} B_{\z} \left( B_{\z}^h G_{\z} B_{\z} \right)^{-\frac{1}{2}}
\end{align*}
as we desired.
\end{proof}

\bibliographystyle{IEEEtran}
\bibliography{refs}

\begin{thebibliography}{10}
\providecommand{\url}[1]{#1}
\csname url@samestyle\endcsname
\providecommand{\newblock}{\relax}
\providecommand{\bibinfo}[2]{#2}
\providecommand{\BIBentrySTDinterwordspacing}{\spaceskip=0pt\relax}
\providecommand{\BIBentryALTinterwordstretchfactor}{4}
\providecommand{\BIBentryALTinterwordspacing}{\spaceskip=\fontdimen2\font plus
\BIBentryALTinterwordstretchfactor\fontdimen3\font minus
  \fontdimen4\font\relax}
\providecommand{\BIBforeignlanguage}[2]{{%
\expandafter\ifx\csname l@#1\endcsname\relax
\typeout{** WARNING: IEEEtran.bst: No hyphenation pattern has been}%
\typeout{** loaded for the language `#1'. Using the pattern for}%
\typeout{** the default language instead.}%
\else
\language=\csname l@#1\endcsname
\fi
#2}}
\providecommand{\BIBdecl}{\relax}
\BIBdecl

\bibitem{comon2010handbook}
P.~Comon and C.~Jutten, \emph{Handbook of Blind Source Separation:
  {Independent} component analysis and applications}.\hskip 1em plus 0.5em
  minus 0.4em\relax Academic press, 2010.

\bibitem{cichocki2002adaptive}
A.~Cichocki and S.~Amari, \emph{Adaptive blind signal and image processing:
  learning algorithms and applications}.\hskip 1em plus 0.5em minus 0.4em\relax
  John Wiley \& Sons, 2002.

\bibitem{lee1998independent}
T.-W. Lee, ``Independent component analysis,'' in \emph{Independent component
  analysis}.\hskip 1em plus 0.5em minus 0.4em\relax Springer, 1998, pp. 27--66.

\bibitem{stone2004independent}
J.~V. Stone, \emph{Independent component analysis: a tutorial
  introduction}.\hskip 1em plus 0.5em minus 0.4em\relax MIT press, 2004.

\bibitem{kim2007}
T.~Kim, H.~T. Attias, S.-Y. Lee, and T.-W. Lee, ``Blind source separation
  exploiting higher-order frequency dependencies,'' \emph{IEEE Trans. Audio,
  Speech, Language Process.}, vol.~15, no.~1, pp. 70--79, 2007.

\bibitem{hiroe2006}
A.~Hiroe, ``Solution of permutation problem in frequency domain {ICA}, using
  multivariate probability density functions,'' in \emph{Proc. ICA}, 2006, pp.
  601--608.

\bibitem{li2009joint}
Y.-O. Li, T.~Adali, W.~Wang, and V.~D. Calhoun, ``Joint blind source separation
  by multiset canonical correlation analysis,'' \emph{IEEE Trans. Signal
  Process.}, vol.~57, no.~10, pp. 3918--3929, 2009.

\bibitem{li2011joint}
X.-L. Li, T.~Adal{\i}, and M.~Anderson, ``Joint blind source separation by
  generalized joint diagonalization of cumulant matrices,'' \emph{Signal
  Processing}, vol.~91, no.~10, pp. 2314--2322, 2011.

\bibitem{anderson2011joint}
M.~Anderson, T.~Adali, and X.-L. Li, ``Joint blind source separation with
  multivariate {Gaussian} model: {Algorithms} and performance analysis,''
  \emph{IEEE Trans. Signal Process.}, vol.~60, no.~4, pp. 1672--1683, 2011.

\bibitem{friedman1974}
J.~H. Friedman and J.~W. Tukey, ``A projection pursuit algorithm for
  exploratory data analysis,'' \emph{IEEE Trans. Comput.}, vol. 100, no.~9, pp.
  881--890, 1974.

\bibitem{huber1985}
P.~J. Huber, ``Projection pursuit,'' \emph{The annals of Statistics}, pp.
  435--475, 1985.

\bibitem{friedman1987exploratory}
J.~H. Friedman, ``Exploratory projection pursuit,'' \emph{Journal of the
  American statistical association}, vol.~82, no. 397, pp. 249--266, 1987.

\bibitem{cardoso1993}
J.-F. Cardoso and A.~Souloumiac, ``Blind beamforming for {non-Gaussian}
  signals,'' in \emph{IEE proceedings F (radar and signal processing)}, vol.
  140, no.~6, 1993, pp. 362--370.

\bibitem{delfosse1995adaptive}
N.~Delfosse and P.~Loubaton, ``Adaptive blind separation of independent
  sources: {A} deflation approach,'' \emph{Signal processing}, vol.~45, no.~1,
  pp. 59--83, 1995.

\bibitem{cruces2004}
S.~A. Cruces-Alvarez, A.~Cichocki, and S.~Amari, ``From blind signal extraction
  to blind instantaneous signal separation: Criteria, algorithms, and
  stability,'' \emph{IEEE Trans. Neural Netw.}, vol.~15, no.~4, pp. 859--873,
  2004.

\bibitem{amari1998adaptive}
S.~Amari and A.~Cichocki, ``Adaptive blind signal processing-neural network
  approaches,'' \emph{Proceedings of the IEEE}, vol.~86, no.~10, pp.
  2026--2048, 1998.

\bibitem{hyvarinen1997fastica}
A.~Hyv{\"a}rinen and E.~Oja, ``A fast fixed-point algorithm for independent
  component analysis,'' \emph{Neural Comput.}, vol.~9, no.~7, pp. 1483--1492,
  1997.

\bibitem{hyvarinen1999fastica}
A.~Hyvarinen, ``Fast and robust fixed-point algorithms for independent
  component analysis,'' \emph{IEEE Trans. Neural Netw.}, vol.~10, no.~3, pp.
  626--634, 1999.

\bibitem{oja2006fastica}
E.~Oja and Z.~Yuan, ``The {FastICA} algorithm revisited: Convergence
  analysis,'' \emph{IEEE Trans. Neural Netw.}, vol.~17, no.~6, pp. 1370--1381,
  2006.

\bibitem{wei2015}
T.~Wei, ``A convergence and asymptotic analysis of the generalized symmetric
  {FastICA} algorithm,'' \emph{IEEE Transl. Signal Process.}, vol.~63, no.~24,
  pp. 6445--6458, 2015.

\bibitem{koldovsky2018ive}
Z.~Koldovsk{\'y} and P.~Tichavsk{\'y}, ``Gradient algorithms for complex
  {non-Gaussian} independent component/vector extraction, question of
  convergence,'' \emph{IEEE Trans. Signal Process.}, vol.~67, no.~4, pp.
  1050--1064, 2018.

\bibitem{koldovsky2017ive}
Z.~Koldovsk{\'y}, P.~Tichavsk{\'y}, and V.~Kautsk{\'y}, ``Orthogonally
  constrained independent component extraction: {Blind} {MPDR} beamforming,''
  in \emph{Proc. EUSIPCO}, 2017, pp. 1155--1159.

\bibitem{jansky2020adaptive}
J.~Jansk{\'y}, J.~M{\'a}lek, J.~{\v{C}}mejla, T.~Kounovsk{\'y},
  Z.~Koldovsk{\'y}, and J.~{\v{Z}}d’{\'a}nsk{\'y}, ``Adaptive blind audio
  source extraction supervised by dominant speaker identification using
  x-vectors,'' in \emph{Proc. ICASSP}, 2020, pp. 676--680.

\bibitem{scheibler2019overiva}
R.~Scheibler and N.~Ono, ``Independent vector analysis with more microphones
  than sources,'' in \emph{Proc. WASPAA}, 2019, pp. 185--189.

\bibitem{scheibler2020fast}
------, ``Fast independent vector extraction by iterative {SINR}
  maximization,'' in \emph{Proc. ICASSP}, 2020, pp. 601--605.

\bibitem{scheibler2020ive}
------, ``{MM} algorithms for joint independent subspace analysis with
  application to blind single and multi-source extraction,''
  \emph{arXiv:2004.03926v1}, 2020.

\bibitem{ike2020overiva}
R.~Ikeshita, T.~Nakatani, and S.~Araki, ``Overdetermined independent vector
  analysis,'' in \emph{Proc. ICASSP}, 2020, pp. 591--595.

\bibitem{lange2016mm}
K.~Lange, \emph{MM optimization algorithms}.\hskip 1em plus 0.5em minus
  0.4em\relax SIAM, 2016, vol. 147.

\bibitem{ono2010auxica}
N.~Ono and S.~Miyabe, ``Auxiliary-function-based independent component analysis
  for super-{Gaussian} sources,'' in \emph{Proc. LVA/ICA}, 2010, pp. 165--172.

\bibitem{ono2011auxiva}
N.~Ono, ``Stable and fast update rules for independent vector analysis based on
  auxiliary function technique,'' in \emph{Proc. WASPAA}, 2011, pp. 189--192.

\bibitem{ono2012auxiva-stereo}
------, ``Fast stereo independent vector analysis and its implementation on
  mobile phone,'' in \emph{Proc. IWAENC}, 2012, pp. 1--4.

\bibitem{ono2018asj}
------, ``Fast algorithm for independent component/vector/low-rank matrix
  analysis with three or more sources,'' in \emph{Proc. ASJ Spring Meeting},
  2018, (in Japanese).

\bibitem{pham2001}
D.-T. Pham and J.-F. Cardoso, ``Blind separation of instantaneous mixtures of
  nonstationary sources,'' \emph{IEEE Trans. Signal Process.}, vol.~49, no.~9,
  pp. 1837--1848, 2001.

\bibitem{degerine2004}
S.~D{\'e}gerine and A.~Za{\"\i}di, ``Separation of an instantaneous mixture of
  {Gaussian} autoregressive sources by the exact maximum likelihood approach,''
  \emph{IEEE Trans. Signal Process.}, vol.~52, no.~6, pp. 1499--1512, 2004.

\bibitem{degerine2006maxdet}
------, ``Determinant maximization of a nonsymmetric matrix with quadratic
  constraints,'' \emph{SIAM J. Optim.}, vol.~17, no.~4, pp. 997--1014, 2006.

\bibitem{yeredor2009HEAD}
A.~Yeredor, ``On hybrid exact-approximate joint diagonalization,'' in
  \emph{Proc. CAMSAP}, 2009, pp. 312--315.

\bibitem{yeredor2012SeDJoCo}
A.~Yeredor, B.~Song, F.~Roemer, and M.~Haardt, ``A “sequentially drilled”
  joint congruence {(SeDJoCo)} transformation with applications in blind source
  separation and multiuser {MIMO} systems,'' \emph{IEEE Trans. Signal
  Process.}, vol.~60, no.~6, pp. 2744--2757, 2012.

\bibitem{tseng2001convergence}
P.~Tseng, ``Convergence of a block coordinate descent method for
  nondifferentiable minimization,'' \emph{Journal of optimization theory and
  applications}, vol. 109, no.~3, pp. 475--494, 2001.

\bibitem{kitamura2016ilrma}
D.~Kitamura, N.~Ono, H.~Sawada, H.~Kameoka, and H.~Saruwatari, ``Determined
  blind source separation unifying independent vector analysis and nonnegative
  matrix factorization,'' \emph{IEEE/ACM Trans. Audio, Speech, Language
  Process.}, vol.~24, no.~9, pp. 1622--1637, 2016.

\bibitem{kameoka2019MVAE}
H.~Kameoka, L.~Li, S.~Inoue, and S.~Makino, ``Supervised determined source
  separation with multichannel variational autoencoder,'' \emph{Neural
  computation}, vol.~31, no.~9, pp. 1891--1914, 2019.

\bibitem{makishima2019independent}
N.~Makishima, S.~Mogami, N.~Takamune, D.~Kitamura, H.~Sumino, S.~Takamichi,
  H.~Saruwatari, and N.~Ono, ``Independent deeply learned matrix analysis for
  determined audio source separation,'' \emph{IEEE/ACM Trans. Audio, Speech,
  Language Process.}, vol.~27, no.~10, pp. 1601--1615, 2019.

\bibitem{sekiguchi2019fast}
K.~Sekiguchi, A.~A. Nugraha, Y.~Bando, and K.~Yoshii, ``Fast multichannel
  source separation based on jointly diagonalizable spatial covariance
  matrices,'' in \emph{Proc. EUSIPCO}, 2019, pp. 1--5.

\bibitem{ito2019fastmnmf}
N.~Ito and T.~Nakatani, ``{FastMNMF}: {Joint} diagonalization based accelerated
  algorithms for multichannel nonnegative matrix factorization,'' in
  \emph{Proc. ICASSP}, 2019, pp. 371--375.

\bibitem{vantrees2004}
H.~L. Van~Trees, \emph{Optimum array processing: {Part IV} of detection,
  estimation, and modulation theory}.\hskip 1em plus 0.5em minus 0.4em\relax
  John Wiley \& Sons, 2004.

\bibitem{warsitz2007maxsnr}
E.~Warsitz and R.~Haeb-Umbach, ``Blind acoustic beamforming based on
  generalized eigenvalue decomposition,'' \emph{IEEE Trans. Audio, Speech,
  Language Process.}, vol.~15, no.~5, pp. 1529--1539, 2007.

\bibitem{ito2017data}
N.~Ito, S.~Araki, and T.~Nakatani, ``Data-driven and physical model-based
  designs of probabilistic spatial dictionary for online meeting diarization
  and adaptive beamforming,'' in \emph{Proc. EUSIPCO}, 2017, pp. 1165--1169.

\bibitem{johnson1992array}
D.~H. Johnson and D.~E. Dudgeon, \emph{Array signal processing: concepts and
  techniques}.\hskip 1em plus 0.5em minus 0.4em\relax Simon \& Schuster, Inc.,
  1992.

\bibitem{warsitz2007blind}
E.~Warsitz and R.~Haeb-Umbach, ``Blind acoustic beamforming based on
  generalized eigenvalue decomposition,'' \emph{IEEE Trans. Audio, Speech,
  Language Process.}, vol.~15, no.~5, pp. 1529--1539, 2007.

\bibitem{khabbazibasmenj2012robust}
A.~Khabbazibasmenj, S.~A. Vorobyov, and A.~Hassanien, ``Robust adaptive
  beamforming based on steering vector estimation with as little as possible
  prior information,'' \emph{IEEE Trans. Signal Process.}, vol.~60, no.~6, pp.
  2974--2987, 2012.

\bibitem{vorobyov2013principles}
S.~A. Vorobyov, ``Principles of minimum variance robust adaptive beamforming
  design,'' \emph{Signal Processing}, vol.~93, no.~12, pp. 3264--3277, 2013.

\bibitem{gomez1998ggd}
E.~G{\'o}mez, M.~Gomez-Viilegas, and J.~Marin, ``A multivariate generalization
  of the power exponential family of distributions,'' \emph{Communications in
  Statistics-Theory and Methods}, vol.~27, no.~3, pp. 589--600, 1998.

\bibitem{amari1996natural-gradient}
S.~Amari, A.~Cichocki, and H.~H. Yang, ``A new learning algorithm for blind
  signal separation,'' in \emph{Proc. NIPS}, 1996, pp. 757--763.

\bibitem{nocedal-Jorge2006optimization}
J.~Nocedal and S.~Wright, \emph{Numerical optimization}.\hskip 1em plus 0.5em
  minus 0.4em\relax Springer Science \& Business Media, 2006.

\bibitem{ike2019ilrma}
R.~Ikeshita, N.~Ito, T.~Nakatani, and H.~Sawada, ``Independent low-rank matrix
  analysis with decorrelation learning,'' in \emph{Proc. WASPAA}, 2019, pp.
  288--292.

\bibitem{murata2001projection-back}
N.~Murata, S.~Ikeda, and A.~Ziehe, ``An approach to blind source separation
  based on temporal structure of speech signals,'' \emph{Neurocomputing},
  vol.~41, no. 1-4, pp. 1--24, 2001.

\bibitem{vincent2006sdr}
E.~Vincent, R.~Gribonval, and C.~F{\'e}votte, ``Performance measurement in
  blind audio source separation,'' \emph{IEEE Trans. Audio, Speech, Language
  Process.}, vol.~14, no.~4, pp. 1462--1469, 2006.

\bibitem{rwcp}
S.~Nakamura, K.~Hiyane, F.~Asano, T.~Nishiura, and T.~Yamada, ``Acoustical
  sound database in real environments for sound scene understanding and
  hands-free speech recognition,'' in \emph{LREC}, 2000.

\bibitem{timit}
J.~Garofolo, L.~Lamel, W.~Fisher, J.~Fiscus, D.~Pallett, N.~Dahlgren, and
  V.~Zue, ``{TIMIT Acoustic-Phonetic Continuous Speech Corpus LDC93S1. Web
  Download. Philadelphia: Linguistic Data Consortium},'' 1993.

\bibitem{chime3}
J.~Barker, R.~Marxer, E.~Vincent, and S.~Watanabe, ``The third ‘{CHiME}’
  speech separation and recognition challenge: Dataset, task and baselines,''
  in \emph{Proc. ASRU}, 2015, pp. 504--511.

\bibitem{atkinson2008}
K.~E. Atkinson, \emph{An introduction to numerical analysis}.\hskip 1em plus
  0.5em minus 0.4em\relax John wiley \& sons, 2008.

\bibitem{nakatani2019wpd-ml}
T.~Nakatani and K.~Kinoshita, ``Maximum likelihood convolutional beamformer for
  simultaneous denoising and dereverberation,'' in \emph{Proc. EUSIPCO}, 2019,
  pp. 1--5.

\end{thebibliography}




\end{document}